\documentclass{article} % For LaTeX2e

\usepackage{fullpage}
\input{preamble.tex}

\title{A Greedy Approach for Budgeted Maximum Inner Product Search}
\date{}
\author{
  Hsiang-Fu Yu\\
  {The University of Texas at Austin} \\
  {rofuyu@cs.utexas.edu}
  \and
  Cho-Jui Hsieh\\
  {The University of California, Davis}\\
  {chohsieh@cs.ucdavis.edu}
  \and
  Qi Lei\\
  {The University of Texas at Austin} \\
  {leiqi@ices.utexas.edu}
  \and
  Inderjit S. Dhillon\\
  {The University of Texas at Austin}\\
  {inderjit@cs.utexas.edu}
}

\def\rofuN{7}
\def\rofuK{3}

\def\Zxshift{-1.75}
\def\Zyshift{-\rofuK}
\def\colorPallete{{"1,0,0", "0,1,0", "0,0,1"}}

\newcommand{\good}[1]{\textcolor{blue}{#1}}

\newcommand{\drawWidx}[3]{
  \foreach \x in {1,...,\rofuK} {    
    \pgfmathparse{\colorPallete[\x-1]};
    \definecolor{curcolor}{rgb}{\pgfmathresult};
    \ifthenelse{\x=1}{\pgfmathparse{int(#1)}}{};
    \ifthenelse{\x=2}{\pgfmathparse{int(#2)}}{};
    \ifthenelse{\x=3}{\pgfmathparse{int(#3)}}{};
    \edef\y{\pgfmathresult}
    \node [draw, regular polygon, regular polygon sides=3, shape border rotate=180, inner sep=0, 
    fill=curcolor!40!white, above= 0cm of rofu\x\y, minimum width=0.5em] () {\scalebox{0.5}{$w_\x$}};
  }
}

\newcommand{\drawHeap}[6] {
    \pgfmathparse{\colorPallete[#1-1]};
    \definecolor{curcolor1}{rgb}{\pgfmathresult};
    \pgfmathparse{\colorPallete[#3-1]};
    \definecolor{curcolor2}{rgb}{\pgfmathresult};
    \pgfmathparse{\colorPallete[#5-1]};
    \definecolor{curcolor3}{rgb}{\pgfmathresult};

    \Tree [.\node[draw,fill=curcolor1!30!white,rectangle,minimum width=2em]{\scriptsize $#2$}; 
      [.\node[draw,fill=curcolor2!30!white,rectangle,minimum width=2em]{\scriptsize$#4$}; ] 
      [.\node[draw,fill=curcolor3!30!white,rectangle,minimum width=2em]{\scriptsize$#6$}; ] 
  ];
}
\newcommand{\brank}[1][]{\def\tst{#1}\ifx\tst\empty\operatorname{\bf rank}\else\operatorname{\bf rank}\left(#1\right)\fi}

\def\gmips{{\sf Greedy-MIPS}\xspace}
\def\pmips{{\sf PCA-MIPS}\xspace}
\def\smips{{\sf Sample-MSIPS}\xspace}
\def\dmips{{\sf Diamond-MSIPS}\xspace}
\def\nmips{{\sf Naive-MIPS}\xspace}
\def\lmips{{\sf LSH-MIPS}\xspace}

\newlength{\leftshift}
\def\yahoo{{\sf yahoo-music}\xspace}
\def\music{{\sf yahoo-music}\xspace}
\def\netflix{{\sf netflix}\xspace}

\begin{document}
\maketitle

\begin{abstract}
  Maximum Inner Product Search (MIPS) is an important task in many machine 
  learning applications such as the prediction phase of a low-rank matrix 
  factorization model for a recommender system. There have been some works on 
  how to perform MIPS in sub-linear time recently. However, most of them do
  not have the flexibility to control the trade-off between search 
  efficient and search quality. In this paper, we study the MIPS problem with a 
  computational budget. By carefully studying the problem structure of MIPS, we 
  develop a novel Greedy-MIPS algorithm, which can handle budgeted MIPS by 
  design. While simple and intuitive, Greedy-MIPS yields surprisingly superior 
  performance compared to state-of-the-art approaches. As a specific 
  example, on a candidate set containing half a million vectors of dimension 200, 
  Greedy-MIPS runs 200x faster than the naive approach while yielding search 
  results with the top-5 precision greater than 75\%.
\end{abstract}

\section{Introduction}
\label{gmips:sec:intro}

In this paper, we study the computational issue in the 
prediction phase for many matrix factorization based latent embedding models 
in recommender systems, which can be mathematically formulated as a Maximum
Inner Product Search (MIPS) problem. 
%Because of the large number of 
%items, the prediction phase for these models usually becomes a 
%Maximum Inner Product Search (MIPS) problem with a very large 
%candidate set.  
%with a very large number of 
%candidates.  
Specifically, given a large collection of $n$ candidate vectors  
\[
\cH = \cbr{\bh_j\in\RR^k: 1 ,\ldots,n}
\]
and a query vector $\bw\in\RR^k$, 
MIPS aims to identify a subset of candidates that have top largest inner product 
values with $\bw$. We also denote $H=[\bh_1,\ldots,\bh_j,\ldots,\bh_n]^\top$ as 
the candidate matrix. A naive linear search procedure to solve MIPS for a 
given query $\bw$ requires $O(nk)$ operations to compute $n$ inner products 
and $O(n \log n)$ operations to obtain the sorted ordering of the $n$ 
candidates.\footnote{When only the  
largest $B$ elements are required, the sorting procedure can be reduced to 
$O(n + B \log B)$ on average using a selection algorithm~\cite{I2A}.} 

%In this paper we focus on the problem of Maximum Inner Product Search (MIPS). 
%In this paper, we focus on the problem of Maximum Inner Product Search (MIPS).
%{\bf Problem Formulation for MIPS.} 
%When $n$ is very large , $O(nk)$ might be 
%identify vectors in $\cH$ with large inner product values with the query 
%vector $\bw$. A nai
%with largest inner product with the query vector: 
%\begin{equation*}
%  \bh^* = \arg\max_{\bh\in \cH} \bh^T \bw. 
%\end{equation*}
%When $n$ is very large, this is a very challenging problem because a naive linear search algorithm takes
%$O(nk)$ time to compute the inner product $\bw^T\bh_j$ for all $j=1, \cdots, n$. 
%How to do this in sub-linear time is thus an interesting and important 
%research topic. 
Recently, MIPS has drawn a lot of attention in the machine learning community.
Matrix factorization (MF) based recommender system~\cite{YK09a,GD12a} is 
one of the most important applications.
%One of the most important applications is matrix factorization based 
%recommendation system.  
In an MF based recommender system, each user $i$ is associated with a vector 
$\bw_i$ of dimension $k$, while each item $j$ is associated with a vector 
$\bh_j$ of dimension $k$. The interaction (such as preference)
between a user and an item is modeled by the value of the inner product 
between $\bw_i$ and $\bh_j$. It is clear that identifying top-ranked items 
in such a system for a user is exactly a MIPS problem.  
Because both the number of users (the number of queries) and the number of 
items (size of vector pool in MIPS) can easily grow to  
millions, a naive linear search is extremely expensive; for example, to 
compute the  preference for all $m$ users over $n$ items with latent embeddings of 
dimension $k$ in a recommender system requires at least $O(mnk)$ operations.
When both $m$ and $n$ are large, the prediction procedure is extremely time 
consuming; it is even {\em slower} than the training procedure used to obtain 
the $m+n$ embeddings, which costs only $O(\abs{\Omega}k)$ operations per iteration. 
Taking the \music dataset as an example, we have 
$m=1M$, $n=0.6M$, $\abs{\Omega} = 250M$, and
\[
  mn = 600B \gg 250M = \abs{\Omega}.
\]
As a result, the development of efficient algorithms for MIPS is needed in 
large-scale recommender systems.  In addition, MIPS can be found in many other 
machine learning applications, such as the prediction  for a multi-class or 
multi-label classifier~\cite{JW10a,HFY14b}, an object detector, a structure 
SVM predicator, and many others.   

There is a recent line of research on accelerating MIPS for large $n$, 
%Recently, there have been some work on how to accelerate MIPS for large $n$,  
such as~\cite{PR12a,NK13b,YB14a,AS14a,BN15a,GB15a}.  
However, most of them do not have the flexibility to control the trade-off 
between search efficiency and  search quality in the prediction phase.
In this paper, we consider the budgeted MIPS problem, which is a generalized version of 
the standard MIPS with a computation budget: {\em how to generate a set of top-ranked 
candidates under a given {\bf budget} on the number of inner products one can 
perform.} By carefully 
studying the problem structure of MIPS, we develop a novel \gmips
algorithm, which handles budgeted MIPS by design. While simple and 
intuitive, Greedy-MIPS yields surprisingly superior performance compared to 
existing state-of-the-art approaches.

{\bf Contributions.} Our contributions can be summarized as follows:
    \begin{itemize}
      \item We carefully study the MIPS problem and develop 
        \gmips, which is a novel algorithm without any 
        nearest neighbor search reduction that is essential in many 
        state-of-the-art approaches~\cite{YB14a,AS14a,BN15a}. 
      \item \gmips is orders of magnitudes  
        faster than many state-of-the-art MIPS approaches to obtain a desired 
        search performance. As a specific example, on the \yahoo data sets with 
        $n=624,961$ and $k=200$, \gmips runs 200x faster 
        than the naive approach and yields search results with the top-5 
        precision more than $75\%$, while the search performance of other state-of-the-art
        approaches under the similar speedup drops to less than $3\%$ precision. 
      \item \gmips supports MIPS with a budget, which brings the ability to 
        control of the trade-off between computation efficiency and 
        search quality in the prediction phase. To the best of our knowledge, 
        among existing MIPS approaches, only the sampling approaches proposed 
        in~\cite{EC97a,GB15a} support the similar flexibility under a limited situation where all 
        the candidates and query vectors are non-negative.
%      \item We are the first to perform empirical comparisons among all the 
%        recently proposed efficient MIPS algorithms including 
%        \pmips~\cite{YB14a}, \lmips~\cite{AS14a,BN15a}, \smips~\cite{GB15a} on 
%        both large real-world and synthetic datasets indicates.   
%      \item \gmips is a greedy approach to identify candidates in a extremely 
%        efficient way. After $O(k)$ per-query pre-processing procedure, \gmips 
%        only costs $O(\log k)$ to generate a new candidate. 
%
%      \item We are the first to perform empirical comparisons among all the 
%        recently proposed efficient MIPS algorithms including 
%        \pmips,\lmips,\smips  on both large real-world and synthetic  datasets indicates  
    \end{itemize}

{\bf Organization.} We first review existing fast MIPS approaches in 
Section~\ref{gmips:sec:related} and introduce the budgeted MIPS problem in 
Section~\ref{gmips:sec:bmips}.  In Section~\ref{gmips:sec:greedy-mips}, 
we propose a novel greedy budgeted MIPS approach called \gmips. We then show the empirical comparison 
in Section~\ref{gmips:sec:exp} and conclude this paper in Section~\ref{gmips:sec:conclusions}.   

\section{Existing Approaches for Fast MIPS}
\label{gmips:sec:related}

Because of its wide applicability, several algorithms have been proposed to design 
efficient algorithms for MIPS. Most of existing approaches consider to reduce 
the MIPS problem to the nearest neighbor search problem (NNS), where the goal 
is to identify the nearest candidates of the given query, and apply an 
existing efficient NNS algorithm to solve the reduced problem~\cite{YB14a,AS14a,BN15a,AS15a,AA16a}.
\cite{YB14a} is the first MIPS work which adopts such a MIPS-to-NNS reduction.
Variants MIPS-to-NNS reduction are also proposed in \cite{AS14a,AS15a}. 
Experimental results in \cite{YB14a} show the superiority of the NNS reduction 
over the traditional branch-and-bound search approaches for MIPS~\cite{PR12a,NK13b}.

Fast MIPS approaches with sampling schemes have become popular 
recently~\cite{EC97a,GB15a}. Various sampling schemes have been proposed to  
handle MIPS problem with different {\em  constraints}. We will briefly review 
two popular sampling schemes in Section~\ref{gmips:sec:sample-mips}. 

    \tikzexternalenable
    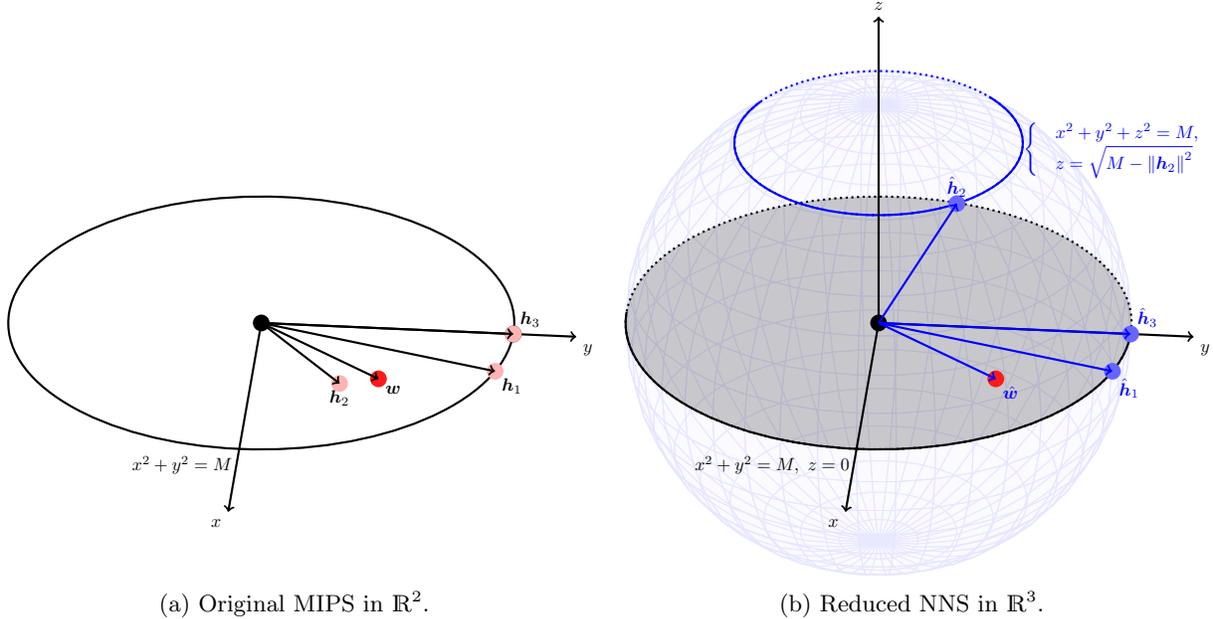
\begin{figure}[t]
      \tdplotsetmaincoords{60}{95}
      \centering
        \begin{subfigure}[b]{0.49\linewidth}
          \begin{resize}{1\linewidth}
            \tikzsetnextfilename{mips-to-nn-1}
            \begin{tikzpicture}[scale=5,tdplot_main_coords]
              \coordinate (O) at (0,0,0);
              \coordinate(W) at (0.4,0.5,0);
              \coordinate(P) at (0.3,0.953939,0);
              \coordinate(Q) at (0.45,0.35,0);
              \coordinate(R) at (0,1,0);

              \coordinate(W2) at (0.4,0.5,1);
              \coordinate(P2) at (0.3,0.953939,0);
              %\coordinate(P2) at (0.4,0.5,0.76811457478);
              \coordinate(Q2) at (0.45,0.35,0.82158383625);
              \coordinate(R2) at (0,1,0);

              \begin{scope}[fill opacity=0.2]
                \pgfsetlinewidth{.1pt}
                \tdplotsphericalsurfaceplot[]{40}{40}{1}{white}{white}{}{}{}%
              \end{scope}

              %\tdplotdrawarc[tdplot_main_coords,dotted,very thick,color=black,fill,fill opacity=0.2]{(0,0,0)}{1}{0}{360}{}{ } 
              \tdplotdrawarc[tdplot_main_coords,very thick,color=black]{(0,0,0)}{1}{-180}{180}{anchor=north east,color=black}{
              $x^2+y^2=M$}

              % draw circle of \bht_2
%              \tdplotdrawarc[tdplot_main_coords,dotted,very thick,color=blue]{(0,0,0.821583836250)}{0.57009}{0}{360}{}{ } 
%              \tdplotdrawarc[tdplot_main_coords,very 
%              thick,color=blue]{(0,0,0.821583836250)}{0.57009}{-120}{135}{}{}
%              \draw node[color=blue,anchor=west] at (0,0.55,0.821583836250) () { 
%                $\begin{cases}
%                  &x^2+y^2+z^2=M,\\
%                  &z= \sqrt{M - \norm{\bh_2}^2}
%              \end{cases}$};

              \draw[very thick,->,black] (0,0,0) -- (1.5,0,0) node[anchor=north east]{$x$};
              \draw[very thick,->,black] (0,0,0) -- (0,1.25,0) node[anchor=north west]{$y$};
              %\draw[very thick,->,black] (0,0,0) -- (0,0,1.5) node[anchor=south]{$z$};

              \draw node[circle,fill] at (O) (){};
              %\draw node[circle,fill] at (1,0,0) (){};
              \draw node[circle,fill] at (0,1,0) (){};
              %\draw node[circle,fill] at (0,0,1) (){};

              \draw node[circle,fill,color=red!90!white] at (W) () {};
              \draw node[circle,fill,color=red!30!white] at (P) () {};
              \draw node[circle,fill,color=red!30!white] at (Q) () {};
              \draw node[circle,fill,color=red!30!white] at (R) () {};
              %\draw node[circle,fill,color=blue!60!white] at (W2) () {};
              %\draw node[circle,fill,color=blue!60!white] at (P2) () {};
              %\draw node[circle,fill,color=blue!60!white] at (Q2) () {};

          %draw figure contents
          %--------------------

              \draw[->,very thick,color=black] (O) -- (W) node[anchor=north west] {$\bw$};
              \draw[->,very thick,color=black] (O) -- (P) node[anchor=north west] {$\bh_1$};
              \draw[->,very thick,color=black] (O) -- (Q) node[anchor=north] {$\bh_2$};
              \draw[->,very thick,color=black] (O) -- (R) node[anchor=south west] {$\bh_3$};
              %\draw[->,very thick,color=blue] (O) -- (W2) node[anchor=south west] {$\bwhat_1$};
              %\draw[->,very thick,color=blue] (O) -- (P2) node[anchor=south west] {$\bhhat_1$};
              %\draw[->,very thick,color=blue] (O) -- (Q2) node[anchor=south] {$\bhhat_2$};
            \end{tikzpicture}
          \end{resize}
          \caption{Original MIPS in $\RR^2$.}
        \label{gmips:fig:nn-reduction-1}
      \end{subfigure}
      \begin{subfigure}[b]{0.49\linewidth}
          \begin{resize}{1\linewidth}
            \tikzsetnextfilename{mips-to-nn-2}
            \begin{tikzpicture}[scale=5,tdplot_main_coords]
              \coordinate (O) at (0,0,0);
              \coordinate(W) at (0.4,0.5,0);
              \coordinate(P) at (0.3,0.953939,0);
              \coordinate(Q) at (0.45,0.35,0);
              \coordinate(R) at (0,1,0);

              \coordinate(W2) at (0.4,0.5,0);
              \coordinate(P2) at (0.3,0.953939,0);
              %\coordinate(P2) at (0.4,0.5,0.76811457478);
              \coordinate(Q2) at (0.45,0.35,0.82158383625);
              \coordinate(R2) at (0,1,0);

              \begin{scope}[fill opacity=0.2]
                \pgfsetlinewidth{.1pt}
                \tdplotsphericalsurfaceplot[]{40}{40}{1}{blue!10!white}{blue!5!white}{}{}{}%
              \end{scope}

              \tdplotdrawarc[tdplot_main_coords,dotted,very thick,color=black,fill,fill opacity=0.2]{(0,0,0)}{1}{0}{360}{}{ } 
              \tdplotdrawarc[tdplot_main_coords,very thick,color=black]{(0,0,0)}{1}{-90}{90}{anchor=north east,color=black}{
              $x^2+y^2=M,\ z=0$}

              % draw circle of \bht_2
              \tdplotdrawarc[tdplot_main_coords,dotted,very thick,color=blue]{(0,0,0.821583836250)}{0.57009}{0}{360}{}{ } 
              \tdplotdrawarc[tdplot_main_coords,very 
              thick,color=blue]{(0,0,0.821583836250)}{0.57009}{-120}{135}{}{}
              \draw node[color=blue,anchor=west] at (0,0.55,0.821583836250) () { 
                $\begin{cases}
                  &x^2+y^2+z^2=M,\\
                  &z= \sqrt{M - \norm{\bh_2}^2}
              \end{cases}$};

              \draw[very thick,->,black] (0,0,0) -- (1.5,0,0) node[anchor=north east]{$x$};
              \draw[very thick,->,black] (0,0,0) -- (0,1.25,0) node[anchor=north west]{$y$};
              \draw[very thick,->,black] (0,0,0) -- (0,0,1.4) node[anchor=south]{$z$};

              \draw node[circle,fill] at (O) (){};
              %\draw node[circle,fill] at (1,0,0) (){};
              %\draw node[circle,fill] at (0,1,0) (){};
              %\draw node[circle,fill] at (0,0,1) (){};

              %\draw node[circle,fill,color=red!90!white] at (W) () {};
              %\draw node[circle,fill,color=red!30!white] at (P) () {};
              %\draw node[circle,fill,color=red!30!white] at (Q) () {};
              %\draw node[circle,fill,color=red!30!white] at (R) () {};
              \draw node[circle,fill,color=red!90!white] at (W2) () {};
              \draw node[circle,fill,color=blue!60!white] at (P2) () {};
              \draw node[circle,fill,color=blue!60!white] at (Q2) () {};
              \draw node[circle,fill,color=blue!60!white] at (R2) () {};

          %draw figure contents
          %--------------------

              %\draw[->,very thick,dotted,color=black] (O) -- (W) node[anchor=north west] {$\bw$};
              %\draw[->,very thick,dotted,color=black] (O) -- (P) node[anchor=north west] {$\bh_1$};
              %\draw[->,very thick,dotted,color=black] (O) -- (Q) node[anchor=north] {$\bh_2$};
              %\draw[->,very thick,dotted,color=black] (O) -- (R) node[anchor=south west] {$\bh_3$};
              \draw[->,very thick,color=blue] (O) -- (W2) node[anchor=north west] {$\bwhat$};
              \draw[->,very thick,color=blue] (O) -- (P2) node[anchor=north west] {$\bhhat_1$};
              \draw[->,very thick,color=blue] (O) -- (Q2) node[anchor=south] {$\bhhat_2$};
              \draw[->,very thick,color=blue] (O) -- (R2) node[anchor=south west] {$\bhhat_3$};
            \end{tikzpicture}
          \end{resize}
          \caption{Reduced NNS in $\RR^3$.}
        \label{gmips:fig:nn-reduction-2}
      \end{subfigure}
      \caption[MIPS-to-NN Reduction.]{MIPS-to-NN reduction. In 
        \ref{gmips:fig:nn-reduction-1}, all the candidate vectors $\cbr{\bh_j}$ and 
        the query vector $\bw$ are in $\RR^{2}$. $\bh_2$ is the nearest 
        neighbor of $\bw$, while $\bh_1$ is the vector yielding the maximum 
        value of the inner product with $\bw$. In \ref{gmips:fig:nn-reduction-2}, 
        the reduction proposed in \cite{YB14a} is applied to $\bw$ and 
        $\cbr{\bh_j}$: $\bwhat = [\bw; 0]^\top$ and 
        $\bhhat_j = [\bh_j; \sqrt{M - \norm{\bh_j}^2}]^\top,\ \forall j$, 
        where $M = \max_j\ \norm{\bh_j}^2$. All the transformed vectors are in
        the 3-dimensional sphere with radius $\sqrt{M}$. 
%        $\RR^3$. In particular, all the transformed candidate vectors are in 
%        the sphere with radius $\sqrt{M}$. 
        As a result, the nearest neighbor 
        of $\bwhat$ in this transformed 3-dimensional NNS problem, $\bhhat_1$, 
        corresponds to the vector $\bh_1$ which  yields the maximum inner 
        product value with $\bw$ in the original 2-dimensional MIPS problem. }
      \label{gmips:fig:nn-reduction}
    \end{figure}
    \tikzexternaldisable

\subsection{Approaches with Nearest Neighbor Search Reduction} 
\label{gmips:sec:nn-mips}
We briefly introduce the concept of the 
reduction proposed in \cite{YB14a}. First, we consider the relationship 
between the Euclidean distance and the inner product:
\begin{align*}
  \norm{\bw - \bh_{j_1}}^2 &= \norm{\bw}^2 + \norm{\bh_{j_1}}^2 - 2 \bw^\top \bh_{j_1}\\
  \norm{\bw - \bh_{j_2}}^2 &= \norm{\bw}^2 + \norm{\bh_{j_2}}^2 - 2 \bw^\top \bh_{j_2}.
\end{align*}
%It is clear that 
When all the candidate vectors $\bh_j$ share the same length; that is,  
\[
  \norm{\bh_{1}} = \norm{\bh_{2}} = \cdots = \norm{\bh_{n}},
\]
the MIPS problem is exactly the same as the NNS problem because 
\begin{equation}
  \norm{\bw - \bh_{j_1}} > \norm{\bw-\bh_{j_2}} \iff \bw^\top\bh_{j_1} < 
  \bw^\top\bh_{j_2}  
  \label{gmips:eq:mips:nn-eq}
\end{equation}
when $\norm{\bh_{j_1}} = \norm{\bh_{j_2}}$. 
However,  when $\norm{\bh_{j_1}}\neq \norm{\bh_{j_2}}$, \eqref{gmips:eq:mips:nn-eq} 
no longer holds. See Figure~\ref{gmips:fig:nn-reduction-1} for an example where 
not all the candidate vectors have the same length. We can see that $\bh_1$ 
is the candidate vector yielding the maximum inner product with $\bw$, while
$\bh_2$ is the nearest neighbor candidate. 
%the 
%candidate vector yielding the maximum inner product value with $\bw$ is not 
%the nearest neighbor candidate, i.e., $\bh_2$ in this example. 

To handle the situation where candidates have different lengths, \cite{YB14a} 
proposes the following transform to reduce the original MIPS problem  with $\cH$ 
and $\bw$ in a $k$ dimensional space to a new NNS problem with $\cHhat=\cbr{\bhhat_1,\ldots,\bhhat_n}$ and 
$\bwhat$ in a $k+1$ dimensional space:
\begin{align}
\bwhat &= \sbr{\bw; 0}^\top, \notag \\
\bhhat_j &= \sbr{\bh_j; \sqrt{M - \norm{\bh_j}^2}}^\top,\ \forall j=1,\ldots, 
n, \label{gmips:eq:mips:nn-reduction}
\end{align}
where $M$ is the maximum squared length over the entire candidate set $\cH$:
\[
  M = \max_{j=1,\ldots,n}\ \norm{\bh_j}^2.
\]
First, we can see that with the above transform, $\norm{\bhhat_j}^2 = M$ for all $j$: 
\[
  \norm{\bhhat_j}^2 = \norm{\bh_j}^2 + M - \norm{\bh_j}^2 = M,\ \forall j.
\]
Then, for any $j_1\neq j_2$, we have 
\begin{align*}
  &\norm{\bwhat - \bhhat_{j_1}}  < \norm{\bwhat - \bhhat_{j_2}} \\
  \iff& M + \norm{\bw}^2 - 2\bw^\top\bh_{j_1} < M + \norm{\bw}^2 - 
  2 \bw^\top\bh_{j_2} \\
  \iff & \bw^\top\bh_{j_1} > \bw^\top\bh_{j_2}.
\end{align*}
With the above relationship, the original $k$-dimensional MIPS problem is equivalent to the 
transformed $k+1$ dimensional NNS problem. In Figure~\ref{gmips:fig:nn-reduction-2}, 
we show the transformed NNS problem for the original MIPS problem presented
in Figure~\ref{gmips:fig:nn-reduction-1}. 

In \cite{AS15a}, another MIPS-to-NNS reduction has been proposed. The high 
level idea is to apply a transformation to $\cH$ such that all the candidate 
vectors {\em roughly} have the same length by appending additional $\kbar$ 
dimensions. In the procedure by \cite{AS15a}, all the $\bh_j$ vectors are 
assumed (or scaled) to have $\norm{\bh_j} \le U,\ \forall j$, where $U<1$ is a positive 
constant. Then the following transform is applied to reduce the original $k$-dimensional 
MIPS problem 
to a new NNS problem with $(k+\kbar)$-dimensional vectors $\cHhat$ and $\bwhat$ 
defined as: 
%of $k+\kbar$ 
%dimension:
\begin{align}
  \bwhat  &= \sbr{\bw; \bzero_{\kbar}}^\top \notag\\
  \bhhat_j &=  \sbr{\bh_j; 1/2-\norm{\bh_j}^{2^1}; 1/2 - 
  \norm{\bh_j}^{2^2};\ldots;1/2-\norm{\bh_j}^{2^{\kbar}}}^\top \label{gmips:eq:mips:nn-reduction-2},
\end{align}
where $\bzero_{\kbar}$ is a zero vector of dimension $\kbar$.
Because $U < 1$, \cite{AS15a} shows that with the transform 
\eqref{gmips:eq:mips:nn-reduction-2}, we have 
$\norm{\bhhat_j}^2 = \kbar/4 + \norm{\bh_j}^{2^{\kbar+1}}$, with the second 
term  vanishing as $\kbar\rightarrow \infty$. Thus, all the candidates $\bhhat_j$ 
approximately have the same length. We can see the idea behind 
\eqref{gmips:eq:mips:nn-reduction-2} is similar to \eqref{gmips:eq:mips:nn-reduction}: 
transforming $\cH$ to $\cHhat$ such that all the candidates have the same length.  
Note that \eqref{gmips:eq:mips:nn-reduction} achieves this goal exactly while 
\eqref{gmips:eq:mips:nn-reduction-2} achieves this goal approximately.  Both 
transforms show a similar empirical performance in \cite{BN15a}.

There are many choices to solve the transformed NNS problem after the 
MIPS-to-NN reduction has been applied. In \cite{AS14a,BN15a,AS15a}, various 
locality sensitive hashing schemes have been considered. In \cite{YB14a}, a 
PCA-tree based approach is proposed, and shows better performance than 
LSH-based approaches, which is consistent to the empirical observations in 
\cite{AA16a} and our experimental results shown in Section~\ref{gmips:sec:exp}.  
In \cite{AA16a}, a simple K-means clustering algorithm is proposed to handled 
the transformed NNS problem.

\subsection{Sampling-based Approaches} 
\label{gmips:sec:sample-mips}

The idea of the sampling-based MIPS approach is first proposed in~\cite{EC97a} as an 
approach to perform approximate matrix-matrix multiplications. 
Its applicability on MIPS problems is studied very recently~\cite{GB15a}. The idea behind
a sampling-based approach called \smips, is about to design an efficient 
sampling procedure such that the $j$-th candidate is selected with probability $p(j)$:
\[
  p(j) \sim \bh_j^\top \bw.
\]
In particular, \smips is an efficient scheme to sample
$(j,t) \in [n]\times [k]$ with the probability $p(j,t)$: 
\[
  p(j,t) \sim h_{jt} w_t.
\]
Each time a pair $(j,t)$ is sampled, we increase the count for the $j$-th item 
by one. By the end of the sampling process, the spectrum of the counts forms an 
estimation of $n$ inner product values. 
Due to the nature of the sampling approach, it can only handle the 
situation where all the candidate vectors and query vectors are {\em nonnegative}. 

\dmips, a diamond sampling scheme proposed in \cite{GB15a}, is an extension of \smips 
to handle the maximum {\bf squared} inner product search problem (MSIPS) where the 
goal is to identify candidate vectors with largest values of $\rbr{\bh_j^\top\bw}^2$. 
If both $\bw$ and $\cH$ are nonnegative or $\bh_j^\top\bw \ge 0,\ \forall j$, 
MSIPS can be used to generate the  solutions for MIPS. However, the solutions 
to MSIPS can be very different from the solutions to MIPS in general. For 
example, if all the inner product values are negative, the ordering for MSIPS 
is the exactly reverse ordering induced by MIPS. Here we can see that the 
applicability of both \smips and \dmips to MIPS is very limited. 

\section{Budgeted Maximum Inner Product Search}
\label{gmips:sec:bmips}

The core idea behind the fast approximate MIPS approaches is to trade the search 
quality for the shorter query latency: the shorter the search latency, the 
lower the search quality. In most existing fast MIPS 
approaches, the trade-off depends on the approach-specific parameters such as 
the depth of the PCA tree in \cite{YB14a} or the number of hash functions in 
\cite{AS14a,BN15a,AS15a}. Such approach-specific parameters are usually 
required to construct approach-specific data structures before any query is 
given, which means that the trade-off is somewhat {\em fixed} for all the 
queries. Particularly, the computation cost for all the query requests is fixed. 
However, in many real-world scenarios, each query request might have a 
different computational budget, which raises the question: 
{\em Can we design a fast MIPS approach which supports the dynamic adjustment 
of the trade-off in the query phase?} 

In this section, we formally define the budgeted MIPS problem which is an 
extension of the standard MIPS problem with a computational budget as a 
parameter given in the query phase. We first summarize the essential 
components for fast MIPS approaches in Section~\ref{gmips:sec:components} and  
give the problem definition of budgeted MIPS in 
Section~\ref{gmips:sec:bmips-def}.

%One
%Unlike the most approaches depending on the MIPS-to-NNS reduction,
%sampling-based MIPS approaches naturally support the control of the 
%{\bf trade-off} between {\em query efficiency} and {\em search quality} in the {\em query 
%phase}. For example, That is sampling-based MIPS approaches can handle a MIPS problem 
%under for any computation budget. 

\subsection{Essential Components for Fast MIPS Approaches}
\label{gmips:sec:components}
Before diving into the details of budgeted MIPS, we first review the 
essential components in fast MIPS algorithms:

\begin{itemize}
  \item {Before any query request:}
    \begin{itemize}
      \item {\em Query-Independent Data Structure Construction:} A 
        pre-processing procedure is performed on the entire candidate sets to 
        construct an approach-specific data structure $\cD$ to store information 
        about $\cH$, such as the LSH hash tables~\cite{AS14a,BN15a,AS15a}, space 
        partition trees (e.g., KD-tree or PCA-tree~\cite{YB14a}), or cluster 
        centroids~\cite{AA16a}.  
    \end{itemize}
  \item {For each query request:}
    \begin{itemize}
      \item {\em Query-dependent {\bf P}re-processing:} In some approaches, a query dependent 
        pre-processing is needed. For example, a vector augmentation is required in  
        all approaches with the MIPS-to-NNS reduction 
        \cite{YB14a,AS14a,BN15a,AA16a}. In addition, \cite{YB14a} also 
        requires another normalization. $T_P$ is used to denote the time 
        complexity of this stage. 
      \item {\em Candidate {\bf S}creening:} In this stage, based on the 
        pre-constructed data structure $\cD$, an efficient procedure is 
        performed to filter candidates such that only a subset of candidates 
        $\cC(\bw)\subset \cH$ is selected.           
        In a naive linear approach, no screening procedure is performed, so 
        $\cC(\bw)$ simply contains all the $n$ candidates. For a tree-based 
        structure, $\cC(\bw)$ contains all the candidates stored in  
        the leaf node of the query vector. In a sampling-based MIPS approach, 
        an efficient sampling scheme is designed to generate highly possible 
        candidates to form $\cC(\bw)$. $T_S$ denotes the computational cost of 
        the screening stage.  

      \item {\em Candidate {\bf R}anking:} An exact ranking is performed on the 
        selected candidates in $\cC(\bw)$ obtained from the screening stage. 
        This involves the computation of $\abs{\cC(\bw)}$ inner products and 
        the sorting procedure among these $\abs{\cC(\bw)}$ values. The overall time 
        complexity $T_R$ is 
        \[
        T_R = O(\abs{\cC(\bw)}k + \abs{\cC(\bw)}\log\abs{\cC(\bw)}). 
        \]
    \end{itemize}
\end{itemize}

The per-query computational cost $T_Q$ is 
\[
  T_Q = T_P + T_S + T_R.
\]
It is clear that the candidate screening stage is the {\em key} component for 
a fast MIPS approach. In terms of the search quality, the performance highly 
depends on whether the screening procedure can identify highly possible 
candidates. In terms of the query latency, the efficiency highly depends on 
the size of $\cC(\bw)$ and how fast to generate $\cC(\bw)$. The major 
difference between various fast MIPS approaches is the choice of the data 
structure $\cD$ and the corresponding screening procedure.  

\subsection{Budgeted MIPS: Problem Definition}
\label{gmips:sec:bmips-def}

Budgeted maximum inner product search is an extension of  the standard 
approximate MIPS problem with a computation budget: how to generate 
top-ranked candidates under a given {\bf budget} on the number of inner 
product operations one can perform.  
Budgeted MIPS has a wide applicability. For example, a real-time 
recommender system must provide a list of recommended items for its users in a 
very short response time. 

Note that the cost for the candidate ranking ($T_R$) is inevitable in the per-query 
cost: $T_Q= T_P + T_S + T_R$. A viable approach to support budgeted MIPS must 
include a screening procedure which satisfies the following requirements:
\begin{itemize}
\item  the flexibility to 
  control the size of $\cC(\bw)$ in the candidate screening stage such that 
  $\abs{\cC(\bw)}\le B$, where $B$ is a given budget, and 
\item an efficient screening procedure to obtain $\cC(\bw)$ in $O(Bk)$ 
  time such that the overall per-query cost is 
\[
  T_Q = O(Bk + B\log B).
\]
\end{itemize}

As mentioned earlier, most recently proposed efficient algorithms such as 
\pmips~\cite{YB14a} and \lmips~\cite{AS14a,BN15a,AS15a} adopt the approach to 
reduce the MIPS problem to an instance of NNS problem, and apply various search 
space partition data structures  
or techniques (e.g., LSH, KD-tree, or PCA-tree) designed for NNS to index the 
candidates $\cH$ in the {\em query-independent pre-processing} stage.
As the construction of $\cD$ is query independent, both the {\bf search performance} 
and the {\bf computation cost} are {\em fixed} when the construction is 
done. For example, the performance of a \pmips depends on the depth 
of the PCA-tree. Given a query vector $\bw$, there is no control to the size 
of $\cC(\bw)$ in the candidate generating phase. LSH-based approaches also have
the similar issue.  As a result, it is not clear how to generalize \pmips and 
\lmips in a principled way to handle the situation with a computational 
budget: how to reduce the size of $\cC(\bw)$ under a limited budget and how to 
improve the performance when a larger budget is given.  

Unlike other NNS-based algorithms, the design of \smips naturally enables it to 
support budgeted MIPS for a nonnegative candidate matrix $H$ and a 
nonnegative query $\bw$. Recall that the core idea behind \smips is to draw a 
sample candidate $j$ among $n$ candidates such that
\[
  p(j) \propto \bh_j^\top \bw.
\]
The more the number of samples, the lower the variance 
of the estimated frequency spectrum. Clearly, \smips has the flexibility to 
control the size of $\cC(\bw)$, and thus is a viable
approach for 
%s a result, \smips can be a viable approach for
the budgeted MIPS problem. However, \smips works only on the situation where 
the entire $\cH$ and $\bw$ are non-negative. \dmips has the similar 
issue.

\section{\gmips: A Novel Approach for Budgeted MIPS}
\label{gmips:sec:greedy-mips}

In this section, we carefully study the problem structure of MIPS and develop 
a simple but novel algorithm called \gmips, which handles budgeted MIPS by design.
Unlike the most recent approaches~\cite{YB14a,AS14a,BN15a,AS15a,AA16a}, \gmips is an 
approach without any reduction to a NNS problem. Moreover, \gmips is a 
viable approach for the budgeted MIPS problem without the non-negativity 
limitation inherited in the sampling approaches.  

As mentioned earlier that the key component for a fast MIPS approach is the 
algorithm used in the candidate screening phase. In budgeted MIPS, 
for any given budget $B$ and query $\bw$, an {\em  ideal procedure} for the 
candidate screening phase costs $O(Bk)$ time to generate $\cC(\bw)$ 
which contains the $B$ items with the largest $B$ inner product values over 
the $n$ candidates in $\cH$. The requirement on the time complexity 
$O(Bk)$ implies that the procedure is independent from $n=\abs{\cH}$, the number of 
candidates in $\cH$. One might wonder whether such an ideal procedure exists 
or not. In fact, designing such an ideal procedure with the requirement to 
{\em generate the largest $B$ items} in $O(Bk)$ time is even more challenging 
than the original budgeted MIPS problem.

\subsection{A Motivating Example for \gmips}
\label{gmips:sec:motivating-example}
Although the existence of an {\em ideal procedure} for a general 
budgeted MIPS problem seems to be impossible, we demonstrate that an {\em 
ideal approach} exists for budgeted MIPS when $k=1$.  
It is not hard to observe that Property~\ref{gmips:prop:two-outcome}
holds for any given $\cH=\cbr{h_1,\ldots,h_n\mid h_j \in \RR}$: 
\begin{property}
For any nonzero query $w \in \RR$ and any 
budget $B > 0$, there are only two possible results for that top $B$ inner 
products between $w$ and $\cH$:
\begin{align*}
  w > 0 &\Rightarrow \text{ Largest  $B$ elements in $\cH$}, \\
  w < 0 &\Rightarrow \text{ Smallest $B$ elements in $\cH$}.
\end{align*}
  \label{gmips:prop:two-outcome}
\end{property}
This property leads to the following simple approach, which is an ideal
procedure for the budgeted MIPS problem when $k=1$:
\begin{itemize}
  \item {\em Query-independent data structure:} a {\bf sorted} list of indices of $\cH$:
    $\mathtt{s[r]},\ \mathtt{r}=1,\ldots,n$ such that $\mathtt{s[r]}$ stores the index 
    to the $r$-th largest candidate. That is 
    \[
      h_{\mathtt{s[1]}} \ge h_{\mathtt{s[2]}} \ge \cdots \ge h_{\mathtt{s[n]}},
    \]
  \item {\em Candidate screening phase:} for any given $w\neq 0$ and $B>0$,
    \[
      \text{\bf return } 
      \begin{cases}
        \text{first $B$ elements: } \cbr{\mathtt{s[1]},\ldots,\mathtt{s[B]}} & \text{ if } w > 0,\\
        \text{last $B$ elements: }\cbr{\mathtt{s[n]},\ldots,\mathtt{s[n-B+1]}} & \text{ if } w < 0
      \end{cases}
    \]
    as the indices of the exact largest-$B$ candidates. 
\end{itemize}
Note that for this simple scenario ($k=1$), neither the {\em query dependent 
pre-processing} nor the {\em  candidate ranking} is needed. Thus, the overall 
time complexity per query is $T_Q = O(B)$. We can see that
Property~\ref{gmips:prop:two-outcome} is the key to the correctness of the 
above procedure. Nevertheless, it is not clear how to generalize 
Property~\ref{gmips:prop:two-outcome} for MIPS problems with  
$k\ge 2$. Fortunately, we can 
directly utilize the fact that Property~\ref{gmips:prop:two-outcome} holds for $k=1$ to 
design an efficient greedy procedure for the candidate screening 
when $k\ge 2$. 
%top-$B$ candidates with based on the marginal ranking $\pi(j| \bw)$ in 
%$o(B k)$ time. In fact, design such an ideal procedure under the time 
%complexity constraint is even harder than the original budgeted MIPS problem.

\subsection{A Greedy Procedure to Candidate Screening}
\label{gmips:sec:gmips-core}
To better describe the idea of the proposed algorithm \gmips, we consider the 
following definition~\eqref{gmips:eq:def-rank}: 
\begin{definition}
  \label{gmips:def:rank}
  The $\brank$ of an item $x$ among a set of items $\cX=\cbr{x_1,\ldots,x_{\abs{\cX}}}$ is defined as 
  \begin{align}
\label{gmips:eq:def-rank}
\brank(x\mid \cX) := \sum_{j=1}^{\abs{\cX}} \Ind{x_j \ge x},
  \end{align}
  where $\Ind{\cdot}$ is the indicator function. A ranking induced by $\cX$ is 
  a function $\pi(\cdot): \cX \rightarrow \cbr{1,\ldots,\abs{\cX}}$  such that 
  %\[
    $\pi(x_j) =  \brank(x_j\mid\cX)\quad \forall x_j \in \cX.$
  %\] %To store a ranking $\pi(\cdot)$, we consider an index array $\mathtt{s}$ 
\end{definition}
One way to store a ranking $\pi(\cdot)$ induced by $\cX$ is by a sorted index array 
$\mathtt{s[r]}$ of size $\abs{\cX}$ such that 
\[
  \pi(x_{\mathtt{s[1]}}) \le \pi(x_{\mathtt{s[2]}})\le \cdots  \le 
  \pi(x_{\mathtt{s[\abs{\cX}]}}).
\]
We can see that $\mathtt{s[r]}$ stores the index to the item $x$ with $\pi(x) = r$. 

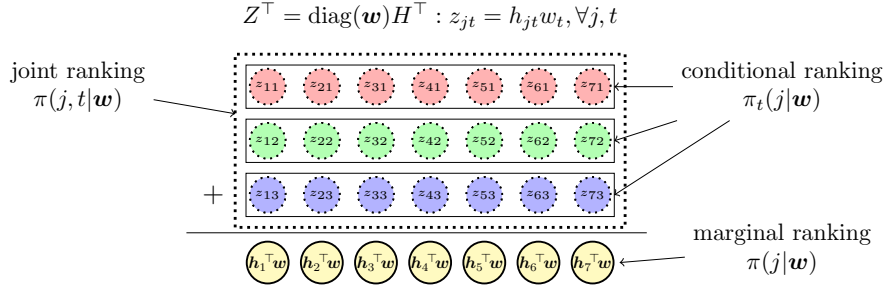
\begin{figure}[t]
  \centering
  \begin{rescale}[0.9]{0.9}
      \begin{tikzpicture}[scale=0.8]
      % draw entire Z
        {
          \draw [dotted,very thick ] (0+\Zxshift-0.1,0+\Zyshift-1.1) rectangle 
          (\rofuN+\Zxshift+0.1,\rofuK+\Zyshift-0.9) {}; 
          \node[align=center] at (0+\Zxshift-3, 0.5*\rofuK + \Zyshift) (JOINT) {
          joint ranking\\ $\pi(j,t|\bw)$};
          \draw [->](JOINT) -- (0+\Zxshift-0.1, 0.5*\rofuK + \Zyshift-0.5);
          \node at (0.5*\rofuN+\Zxshift,-0.2){ $Z^\top = \diag(\bw) H^\top: z_{jt}=h_{jt}w_t, \forall j,t$};
          \foreach \x in {1,...,\rofuK} {    
            \pgfmathparse{\colorPallete[\x-1]};
            \definecolor{curcolor}{rgb}{\pgfmathresult};
            \foreach \y in {1,...,\rofuN} {
              \node[circle,inner sep=0,fill=none,draw,thick, dotted,minimum 
              width=1.5em] at (\y-0.5+\Zxshift,\rofuK-\x-0.5 +\Zyshift) {\tiny {$z_{\y\x}$}};
            }
          %\node at (\rofuN+0.5,\rofuK-\x+0.5) {$\small{\bhbar_\x^\top$}};
          }
          % horizontal line with plus sign
          \draw (\Zxshift-1, \Zyshift-1.2) -- (\Zxshift +\rofuN+0.5, \Zyshift-1.2);
          \node at (\Zxshift-0.5, -0.5 + \Zyshift) {\large{$+$}};
        }
        {
          \foreach \x in {1,...,\rofuK} {    
            \pgfmathparse{\colorPallete[\x-1]};
            \definecolor{curcolor}{rgb}{\pgfmathresult};
            \draw[draw,fill=none] (0+\Zxshift+0.1,\rofuK-\x+\Zyshift-1+0.1) rectangle 
            (\rofuN+\Zxshift-0.1,\rofuK-\x+\Zyshift-0.1) ; 
            \foreach \y in {1,...,\rofuN} {
              \node[circle,inner sep=0,fill=curcolor!30!white,draw,thick, dotted,minimum 
              width=1.5em] at (\y-0.5+\Zxshift,\rofuK-\x-0.5 +\Zyshift) {\tiny {$z_{\y\x}$}};
            }
          }
          \node[align=center] at (\rofuN+\Zxshift+3, 0.5*\rofuK + \Zyshift) (COND) 
          { conditional ranking\\  $\pi_t(j|\bw)$};
          \draw[->] (COND) -- (\rofuN+\Zxshift-0.1, 0.5*\rofuK + \Zyshift-0);
          \draw[->] (COND) -- (\rofuN+\Zxshift-0.1, 0.5*\rofuK + \Zyshift-1);
          \draw[->] (COND) -- (\rofuN+\Zxshift-0.1, 0.5*\rofuK + \Zyshift-2);
        }
        \foreach \y in {1,...,\rofuN} {
          \node[circle,draw,inner sep=0,fill=yellow!30!white,thick,minimum 
          width=1.5em] at (\y-0.5+\Zxshift,-1.75 +\Zyshift) {\scalebox{0.7}{$\bh_{\y}\!^\top\!\bw$}};
          %\draw[dotted] (\y-1, \rofuK-\x) rectangle (\y, \rofuK-\x+1);
        }
        \node[align=center] at (\rofuN+\Zxshift+3,-\rofuK+\Zyshift+1.5) 
        (MARGINAL){ marginal ranking \\ $\pi(j|\bw)$} ;
        \draw[->] (MARGINAL) -- (\rofuN+\Zxshift, -1.75+\Zyshift) ;
      \end{tikzpicture}
    \end{rescale}
    \caption[Three rankings considered in \gmips]{$nk$ multiplications in a naive linear MIPS approach.}
  \label{gmips:fig:mips-nk-ops}
\end{figure}

In order to design an efficient candidate screening procedure, we 
carefully study the operations required for MIPS. In the naive linear MIPS 
approach, $nk$ multiplication operations are required to  obtain $n$ inner 
product values $\cbr{\bh_1^\top \bw,\ldots,\bh_n^\top\bw}$. To understand and 
analyze the computation required for MIPS, we define
%$nk$ operations form 
an {\em implicit matrix} $Z \in \RR^{n\times k}$:
\[
Z = H\diag(\bw),  
\]
where $\diag(\bw) \in \RR^{k\times k}$ is a matrix with $\bw$ as it diagonal. 
The $(j,t)$ entry of $Z$ denotes the multiplication operation 
$z_{jt} = h_{jt}w_t$ and $\bz_j = \diag(\bw)\bh_j$ denotes the $j$-th row of 
$Z$. In Figure~\ref{gmips:fig:mips-nk-ops}, we use $Z^\top$ to demonstrate the 
implicit matrix.  The implicit matrix $Z$ is query dependant, that is, 
the values of $Z$ depend on the query vector $\bw$. 
Note that $n$ inner product values can be obtained by taking the 
column-wise summation of $Z^\top$. In particular, we have
\[
  \bh_j^\top\bw = \sum_{t=1}^k z_{jt},\ j=1,\ldots,n. 
\]
Thus, the ranking induced by the $n$ inner product values can be characterized 
by the {\em  marginal ranking} $\pi(j| \bw)$ defined on the implicit matrix $Z$ 
as follows: 
\begin{align}
  \pi(j| \bw)
  &:= \brank\rbr{\left.\sum_{t=1}^{k}z_{jt} \ \right| \
    \cbr{\sum_{t=1}^{k}z_{1t},\cdots,\sum_{t=1}^{k} z_{nt}}} \label{gmips:eq:marginal-rank}\\
  &= \brank\rbr{\bh_j^\top \bw \mid \cbr{\bh_{1}^\top 
  \bw,\ldots,\bh_n^\top\bw}}.  \nonumber
\end{align}
As mentioned earlier, it is hard to design an ideal candidate screening
procedure which generates $\cC(\bw)$ based on the marginal ranking. Because 
the main goal for the candidate screening phase is to quickly identify 
candidates which are highly possible to be top-ranked items, it suffices to 
have an efficient procedure generating $\cC(\bw)$ by an approximation 
ranking. Here we propose a greedy heuristic ranking:
\begin{align}
  \pibar(j|\bw) := \brank\rbr{\left.\max_{t=1}^{k}z_{jt} \ \right| \
    \cbr{\max_{t=1}^{k}z_{1t},\cdots,\max_{t=1}^{k} z_{nt}}},
    \label{gmips:eq:greedy-rank}
\end{align}
which is obtained by replacing the summation terms in 
\eqref{gmips:eq:marginal-rank} by ${\max}$ operators. The intuition behind 
this heuristic is that the largest element of $\bz_j$ multiplied by $k$ is an upper 
bound of $\bh_j^\top \bw$: 
\[
  \bh_j^\top \bw = \sum_{t=1}^{k}z_{jt} \le k \rbr{\max_{t=1}^k z_{jt}}. 
\]
Thus, $\pibar(j|\bw)$, which is induced by such an upper bound of 
$\bh_j^\top\bw$, could be a reasonable approximation ranking for the marginal 
ranking $\pi(j|\bw)$.   

Next we design an efficient procedure which generates $\cC(\bw)$ according to 
the ranking $\pibar(j|\bw)$ defined in \eqref{gmips:eq:greedy-rank}. 
First, based on the relative orderings of $\cbr{z_{jt}}$, we consider the 
{\em joint ranking} and the {\em conditional ranking} defined as follows: 
\begin{itemize}
%  \item Marginal ranking: $\pi(j| \bw)$ is the exact ranking over the $n$ inner products.
%    \begin{align*}
%      \pi(j| \bw)
%      &:= \brank\rbr{\sum_{t}z_{jt} \mid 
%      \cbr{\sum_{t}z_{1t},\ldots,\sum_{t}z_{nt}}} \\
%      &= \brank\rbr{\bh_j^\top \bw \mid \cbr{\bh_{1}^\top 
%      \bw,\ldots,\bh_n^\top\bw}}.
%    \end{align*}
  \item Joint ranking: $\pi(j,t| \bw)$ is the exact ranking over the $nk$ 
    entries of $Z$. 
    \[
      \pi(j,t| \bw):= \brank (z_{jt} \mid \cbr{z_{11},\ldots,z_{nk}}).
    \]
  \item Conditional ranking:  $\pi_t(j| \bw)$ is the exact  ranking over the $n$ 
    entires of the $t$-th row of $Z^\top$. 
    \[
      \pi_t(j| \bw) := \brank(z_{jt} \mid  \cbr{z_{1t},\ldots,z_{nt}}). 
    \]
\end{itemize}
See Figure~\ref{gmips:fig:mips-nk-ops} for an illustration for both rankings. 
Similar to the marginal ranking, both joint and conditional rankings 
are query dependent. 

Observe that, in \eqref{gmips:eq:greedy-rank}, for each $j$, only a single 
maximum entry of $Z$, $\max_{t=1}^k z_{jt}$, is considered to obtain the 
ranking $\pibar(j|\bw)$. To generate $\cC(\bw)$ based on $\pibar(j|\bw)$, we 
can iterate $(j,t)$ entries of $Z$ in a greedy sequence such that $(j_1,t_1)$ 
is visited before $(j_2,t_2)$ if $z_{j_1t_1} > z_{j_2t_2}$, which is exactly 
the sequence corresponding to the joint ranking $\pi(j,t|\bw)$. Each time an
entry $(j,t)$ is visited, we can include the index $j$ into $\cC(\bw)$ if 
$j \notin \cC(\bw)$. In Theorem~\ref{gmips:thm:joint-greedy}, we show that the sequence to 
include a newly observed $j$ into $\cC(\bw)$ is exactly the sequence induced by
the ranking $\pibar(j|\bw)$ defined in~\eqref{gmips:eq:greedy-rank}.  

\begin{theorem}
  For all $j_1$ and $j_2$ such that $\pibar(j_1|\bw) < \pibar(j_2|\bw)$,  
  $j_1$ will be included into $\cC(\bw)$  before $j_2$ if we iterate $(j,t)$ pairs 
  following the sequence induced by the joint ranking $\pi(j,t|\bw)$.  
  \label{gmips:thm:joint-greedy}
\end{theorem}
\begin{proof}
  Let $t_1 = \arg\max_{t=1}^k z_{j_1t}$ and $t_2 = \arg\max_{t=1}^k z_{j_2t}$. 
  By the definition of $t_1$, we have  
  $\pi(j_1,t_1|\bw) < \pi(j_1,t|\bw),\ \forall t\neq t_1$. Thus, $(j_1,t_1)$ 
  will be first entry among $\cbr{(j_1,1),\ldots,(j_1,k)}$ to be visited in 
  the sequence corresponding to the joint ranking $\pi(j,t|\bw)$. Similarly, 
  $(j_2,t_2)$ will be the first visited entry among $\cbr{(j_2,1,),\ldots,(j_2,k)}$.  
  We also have 
  \begin{align*}
    \pibar(j_1|\bw) < \pibar(j_2|\bw) \Rightarrow z_{j_1t_1} > z_{j_2t_2} 
    \Rightarrow \pi(j_1,t_1|\bw) < \pi(j_2,t_2|\bw). 
  \end{align*}
  Thus, $j_1$ will be included into $\cC(\bw)$ before $j_2$. 
\end{proof}

At first glance, generating $(j,t)$ in the sequence according to the joint ranking 
$\pi(j,t|\bw)$ might require the access to all the $nk$ entries of $Z$ and cost 
$O(nk)$ time. In fact, based on Property~\ref{gmips:prop:rank-inv} of 
conditional rankings, we can design an efficient variant of the $k$-way merge 
algorithm~\cite[Chapter 5.4.1]{TAOCP3} to generate $(j,t)$ pairs in the 
desired sequence iteratively. 

\begin{property}
  \label{gmips:prop:rank-inv}
  Given a fixed candidate matrix $H$, for any possible $\bw$ with $w_t\neq 0$, 
  the conditional ranking $\pi_t(j| \bw)$ is either $\pi_{t+}(j)$ or 
  $\pi_{t-}(j)$:
  \begin{itemize}
    \item $\pi_{t+}(j) = \brank(h_{jt}\mid \cbr{h_{1t},\ldots,h_{nt}})$,
    \item $\pi_{t-}(j) = \brank(-h_{jt}\mid \cbr{-h_{1t},\ldots,-h_{nt}})$.
  \end{itemize}
  In particular, we have 
  \begin{align*}
    \pi_t(j|\bw) =
    \begin{cases}
      \pi_{t+}(j) & \text{ if } w_t > 0,\\
      \pi_{t-}(j) & \text{ if } w_t < 0.
    \end{cases}
  \end{align*}
\end{property}

Similar to Property~\ref{gmips:prop:two-outcome}, 
Property~\ref{gmips:prop:rank-inv} enables us to  
characterize a {\em query dependent} conditional ranking $\pi_t(j|\bw)$ by two {\em 
query independent} rankings $\pi_{t+}(j)$ and ${\pi_{t-}(j)}$. As a result,  
similar to the motivating example in 
Section~\ref{gmips:sec:motivating-example},for each $t$, we can construct and store a {\bf 
sorted} index array $\mathtt{s_t[r]},\ r=1,\ldots,n$ such that 
\begin{align}
\label{gmips:eq:sorted-idx}
  \pi_{t+}(\mathtt{s_t[1]}) \le \pi_{t+}(\mathtt{s_t[2]}) \le \cdots \le 
  \pi_{t+}(\mathtt{s_t[n]}),
\end{align}
or equivalently 
\begin{align}
\label{gmips:eq:sorted-idx-2}
  \pi_{t-}(\mathtt{s_t[1]}) \ge \pi_{t-}(\mathtt{s_t[2]}) \ge \cdots \ge 
  \pi_{t-}(\mathtt{s_t[n]}). 
\end{align}
Thus, in the phase of {\em query-independent data structure 
construction} of \gmips, we compute and store {\em query-independent} rankings 
$\pi_{t+}(\cdot), \ t=1,\ldots,k$ by $k$ sorted index arrays of length $n$: 
$\mathtt{s_t[r]},\ r=1,\ldots,n,\ t=1,\ldots,n$ 
such that \eqref{gmips:eq:sorted-idx} holds. The entire construction costs
$O(k n \log n)$ time and $O(kn)$ space. 

\begin{algorithm}[t]
  \caption{{\tt ConditionalIterator}: an iterator iterates 
    $j \in \cbr{1,\ldots,n}$ based on the conditional ranking $\pi_t(j|\bw)$. 
    This pseudo code assumes that the $k$ sorted index arrays 
  $\mathtt{s_t[r]},\ r=1,\ldots,n,\ t=1,\ldots,k$ are available.} 
  \label{gmips:alg:cond-iter}
      \begin{itemize}[leftmargin=0em]
        %\item {Iterator for each conditional ranking $\pi_t(j|\bw)$:}
        \item[] {\tt{{\bf class} ConditionalIterator:}}  
      \begin{itemize}[leftmargin=1.5em]
        %\item[] {\tt {\bf def} \underline{\hspace{1em}}init\underline{\hspace{1em}}(dim\_idx, query\_val):}
        \item[] {\tt {\bf def} constructor(dim\_idx, query\_val):}
          \begin{itemize}[leftmargin=1em]
            \item[] {\tt t, w, ptr $\leftarrow$ dim\_idx, query\_val, 1} 
            %\item[] {\tt ptr = 1}
          \end{itemize}
        \item[] {\tt {\bf def} current():}
          %\begin{itemize}[leftmargin=1.5em]
            %\item[] 
              {\tt {\bf return}} $\begin{cases}
                \mathtt{s_t[ptr]}& \text{if } \mathtt{w > 0}, \\
                \mathtt{s_t[n-ptr+1]}& \text{otherwise. } 
              \end{cases}$
          %\end{itemize}
        \item[] {\tt {\bf def} hasNext():} 
          %\begin{itemize}[leftmargin=1.5em] %\item [] 
              {\tt{\bf return} $(\mathtt{ptr} < n)$}
          %\end{itemize}
        \item[] {\tt {\bf def} getNext():}
          %\begin{itemize}[leftmargin=1.5em]
            %\item[] 
             $\mathtt{ptr\leftarrow ptr+1}$ and 
            %\item[] 
             {\tt {\bf return} current()}
          %\end{itemize}
      \end{itemize}
    \end{itemize}
      \begin{comment}
      \begin{itemize}[leftmargin=1em]
        \item[] {\tt{{\bf class} ConditionalIterator:}}
          \begin{itemize}[leftmargin=1em]
            \item[] {\tt {\bf def} \underline{\hspace{1em}}init\underline{\hspace{1em}}(dim\_idx, query\_val):}
              \begin{itemize}[leftmargin=1em]
                \item[] {\tt t, w, ptr $=$ dim\_idx, query\_val, 1} 
            %\item[] {\tt ptr = 1}
              \end{itemize}
            \item[] {\tt {\bf def} current():}
              \begin{itemize}[leftmargin=1em]
                \item[] 
                  {\tt {\bf return}} $\begin{cases}
                    \mathtt{s_t[ptr]}& \text{if } \mathtt{w > 0}, \\
                    \mathtt{s_t[n-ptr+1]}& \text{otherwise. } 
                  \end{cases}$
              \end{itemize}
            \item[] {\tt {\bf def} hasNext():} 
              \begin{itemize}[leftmargin=1em]
                \item [] {\tt{\bf return} $\mathtt{(ptr < n)}$}
              \end{itemize}
            \item[] {\tt {\bf def} next():}
              \begin{itemize}[leftmargin=1em]
                \item[] $\mathtt{ptr = ptr + 1}$
                \item[] {\tt {\bf return} current()}
              \end{itemize}
          \end{itemize}
      \end{itemize}
      \end{comment}
\end{algorithm}

Next we describe the details of the proposed \gmips algorithm when a query 
$\bw$ and the budget $B$ are given. As mentioned earlier, \gmips utilizes the idea of the $k$-way 
merge algorithm to visit $(j,t)$ entries of $Z$ 
according to the joint ranking $\pi(j,t|\bw)$. Designed to merge $k$ 
{\bf sorted} sublists into a single sorted list, the $k$-way 
merge algorithm uses 1) $k$ pointers, one for each sorted sublist, and 2) a 
binary tree structure (either a heap or a selection tree) containing the elements pointed by these $k$ pointers
to obtain the next element to be appended into the sorted list~\cite[Chapter  
5.4.1.]{TAOCP3}. 

\subsubsection{Query-dependent Pre-processing}
\label{gmips:sec:query-dep-pre-proc}
In \gmips, we divide $nk$ entries of $(j,t)$ into $k$ groups. The $t$-th group 
contains $n$ entries: 
\[
  \cbr{(j,t): j=1,\ldots,n}.
\]
Here we need an {\em iterator} playing a similar role as the pointer which can iterate index 
$j\in \cbr{1,\ldots,n}$ in the {\em sorted} sequence induced by the 
conditional ranking $\pi_t(\cdot|\bw)$. Utilizing 
Property~\ref{gmips:prop:rank-inv}, the $t$-th pre-computed sorted arrays 
$\mathtt{s_t[r]},\ r=1,\ldots,n$ can be used to construct such an iterator, 
called {\tt ConditionalIterator}, which iterates an index $j$ one by one in 
the desired sorted sequence. {\tt ConditionalIterator} needs to support {\tt 
current()} to access the currently pointed index $j$ and {\tt 
getNext()} to advance the iterator.  In Algorithm~\ref{gmips:alg:cond-iter}, 
we describe a pseudo code for {\tt ConditionalIterator}, which utilizes the 
facts \eqref{gmips:eq:sorted-idx} and \eqref{gmips:eq:sorted-idx-2} such that 
both the construction and the index access cost $O(1)$ space and $O(1)$ time.   
For each $t$, we use $\mathtt{iters}[t]$ to denote the {\tt 
ConditionalIterator } for the $t$-th conditional ranking $\pi_t(j|\bw)$.  

\begin{algorithm}[t]
  \caption{Query-dependent pre-processing procedure in \gmips.}
  \label{gmips:alg:query-dep-proc}
     \begin{itemize}
       \item {\bf Input:} query $\bw \in \RR^k$
       \item For ${t}=1,\ldots,k$
         \begin{itemize}
           \item $\mathtt{iters}[t] \leftarrow \mathtt{ConditionalIterator}(t,\ w_{{t}}\mathtt{)}$
           \item $j \leftarrow \mathtt{iters}[t].\mathtt{current()}$
           \item $z \leftarrow h_{jt}w_t$
           \item $\mathtt{Q.push}\rbr{\rbr{z, t}}$
         \end{itemize}
       \item {\bf Output:}
         \begin{itemize}
           \item $\mathtt{iters[t]},\ t=1,\ldots,k$: iterators for conditional ranking $\pi_t(\cdot|\bw)$. 
           \item $\mathtt{Q}$: a max-heap containing $\cbr{(z,t)\mid z = \max_{j=1}^n z_{jt},\ t=1,\ldots,k}$.
         \end{itemize}
     \end{itemize}
\end{algorithm}

Regarding the binary tree structure used in \gmips, we consider a max-heap $\mathtt{Q}$ of 
$(z,t)$ pairs. $z\in\RR$ is the {\em compared key} used to maintain the heap property 
of $\mathtt{Q}$, and $t\in\cbr{1,\ldots,k}$ is an integer to denote the index 
to a entry group. Each $(z,t) \in \mathtt{Q}$ denotes the $(j,t)$ entry of $Z$ where 
\[
  {j=\mathtt{iters}[t].\mathtt{current()}\quad\text{and}\quad z = z_{jt} = 
  h_{jt}w_t}. 
\]
Note that there are most $k$ elements in the max-heap at any time. Thus, we 
can implement $\mathtt{Q}$ by a binary heap such that it supports 
\begin{itemize}
  \item $\mathtt{Q.top()}$: returns the maximum pair $(z,t)$ of $\mathtt{Q}$ in $O(1)$ time, 
  \item $\mathtt{Q.pop()}$: deletes the maximum pair of $\mathtt{Q}$ in $O(\log k)$ time, and 
  \item $\mathtt{Q.push}\rbr{\rbr{z,t}}$: inserts a new pair in $O(\log k)$ time.  
\end{itemize}
Note that the entire \gmips can also be implemented using a selection tree 
among the $k$ entries pointed by the $k$ iterators. For the simplicity of 
presentation, we use a max-heap to describe the idea of \gmips first and 
describe the details of \gmips with a selection tree in the end 
of Section~\ref{gmips:sec:screening}.   

In the {\em  query-dependent pre-processing} phase of \gmips, we need to 
construct $\mathtt{iters}[t],\ t=1,\ldots,k$, one for each conditional ranking 
$\pi_t(j|\bw)$, and a max-heap $\mathtt{Q}$ which is 
initialized to contain 
\[
\cbr{(z,t)\mid z = \max_{j=1}^n z_{jt},\ t=1,\ldots,k}. 
\]
A detailed procedure is described in Algorithm~\ref{gmips:alg:query-dep-proc}, 
which costs $O(k\log k)$ time and $O(k)$ space.  

\subsubsection{Candidate Screening} %Efficient procedure to generate $(j,t)$ according to $\pi(j,t)$.
\label{gmips:sec:screening}
Recall the requirements for a viable candidate screening procedure to support 
budgeted MIPS: 1) the flexibility to control the size $\abs{\cC(\bw)}\le B$; and 2) an 
efficient procedure runs in $O(Bk)$. The core idea of \gmips is to iteratively traverse
$(j,t)$ entries of $Z$ in a {\em greedy} sequence and collect newly observed 
indices $j$ into $\cC(\bw)$ until $\abs{\cC(\bw)}=B$. In particular, if 
$r=\pi(j,t|\bw)$, then $(j,t)$ entry is visited at the $r$-th iterate. 
Utilizing the max-heap $\mathtt{Q}$ and the $k$ iterators: $\mathtt{iters}[t]$, 
we can design an iterator, called {\tt JointIterator}, which iterates $(j,t)$ 
pairs one by one in the desired greedy sequence induced by joint ranking 
$\pi(j,t|\bw)$. Following the $k$-way merge algorithm, in 
Algorithm~\ref{gmips:alg:joint-iter}, we describe a detailed pseudo code for 
such an iterator. {\tt JointIterator} costs $O(k\log k)$ 
time to run Algorithm~\ref{gmips:alg:query-dep-proc} to construct and 
initialize $\mathtt{Q}$ and $\mathtt{iters}[t]$, and costs $O(\log k)$ 
time to advance to the next entry. In Algorithm~\ref{gmips:alg:screen}, we 
describe our first {\em candidate screening} procedure with a budget 
$B$ for \gmips, which is a simple {\bf while}-loop to iterate $(j,t)$ entries 
using the {\tt JointIterator} with $\bw$ until $\abs{\cC(\bw)} = B$. 

\begin{algorithm}
  \caption{{\tt JointIterator:} an iterator generates $(j,t)$ pairs 
  one by one based on the joint ranking $\pi(j,t|\bw)$. The constructor costs 
  $O(k\log k)$ time to build a max-heap $\mathtt{Q}$. The time complexity to 
generate a pair is $O(\log k)$. }
  \label{gmips:alg:joint-iter}
  \begin{itemize}[leftmargin=0em]
   \item[] {\tt{{\bf class} JointIterator:}}  
      \begin{itemize}[leftmargin=1.5em]
        %\item[] {\tt {\bf def} \underline{\hspace{1em}}init\underline{\hspace{1em}}(dim\_idx, query\_val):}
        \item[] {\tt {\bf def} constructor}$(\bw)$: \hfill $\cdots O(k\log k)$
          \begin{itemize}[leftmargin=1em]
            \item[] Run Algorithm~\ref{gmips:alg:query-dep-proc} with $\bw$ to initialize
              $\mathtt{Q}$ and $\mathtt{iters}[t],\ t=1,\ldots,k$
            \item[]$\mathtt{ptr\leftarrow 1}$.
          \end{itemize}
        \item[] {\tt {\bf def} current():} \hfill $\cdots O(1)$
          \begin{itemize}[leftmargin=1.5em]
            \item[] $(z,t) \leftarrow \mathtt{Q.top()}$ 
            \item[] $j \leftarrow \mathtt{iters}[t]\mathtt{.current()}$
            \item[] {\bf return }$(j,t)$
          \end{itemize}
        \item[] {\tt {\bf def} hasNext():} 
          {\tt{\bf return} $(\mathtt{ptr} < nk)$} \hfill $\cdots O(1)$
        \item[] {\tt {\bf def} getNext():}\hfill $\cdots O(\log k)$
          \begin{itemize}[leftmargin=1.5em]
            \item[] $(z,t) \leftarrow \mathtt{Q.pop()}$ \hfill $\cdots O(\log k)$
            \item[] {\bf if} $\mathtt{iters}[t]\mathtt{.hasNext()}$:
              \begin{itemize}[leftmargin=1.5em]
                \item[] $j \leftarrow \mathtt{iters}[t]\mathtt{.getNext()}$
                \item[] $z\leftarrow h_{jt} w_t$
                \item[] $\mathtt{Q.push}\rbr{\rbr{z,t}}$ \hfill $\cdots O(\log k)$
              \end{itemize}
            \item[] $\mathtt{ptr\leftarrow ptr+1}$
            \item[] {\tt {\bf return} current()}
          \end{itemize}
      \end{itemize}
    \end{itemize}
\end{algorithm}
\begin{algorithm}
  \caption{Candidate screening procedure in \gmips.}
  \label{gmips:alg:screen}
  \begin{itemize}
    \item {\bf Input:} $\bw$ and an empty $\cC(\bw)$
    \item $\mathtt{jointIter} \leftarrow \mathtt{JointIterator}(\bw)$ \hfill $\cdots O(k\log k)$
    \item $(j,t) \leftarrow \mathtt{jointIter.current()}$
    \item {\bf while} $\abs{\cC(\bw)} < B$:
      \begin{itemize}
        \item {\bf if} $j \notin \cC(\bw)$: append $j$ to $\cC(\bw)$
        \item $(j,t)\leftarrow \mathtt{jointIter.getNext()}$ \hfill $\cdots O(\log k)$
      \end{itemize}
    \item {\bf Output:} $\cC(\bw)=\cbr{j\mid \pibar(j|\bw) \le B}$ 
  \end{itemize}
\end{algorithm}

To analyze the time complexity of Algorithm~\ref{gmips:alg:screen}, we need to 
know the number of the iterations of the {\bf while}-loop before the stop 
condition is satisfied. The following Theorem~\ref{gmips:thm:bounded-iter} 
gives an upper bound on this number of iterations.  
\begin{theorem}
  \label{gmips:thm:bounded-iter}
  There are at least $B$ distinct indices $j$ in the first $Bk$ entries $(j,t)$ 
  in terms of the joint ranking $\pi(j,t|\bw)$ for any $\bw$; that is,  
  \begin{align}
    \abs{\cbr{j \mid \forall (j,t) \text{ such that } \pi(j,t|\bw) \le Bk}} 
    \ge B.  \label{gmips:eq:bounded-iter}
  \end{align}
\end{theorem}
\begin{proof}
  By grouping these first $Bk$ entries by the index $t$ and applying the pigeonhole 
  principle, we know that there exists a group $G$ such that it contains at least 
  $B$ entries. Because each entry in the same group has a distinct $j$ index, 
  we know that the group $G$ contains at least $B$ distinct indices $j$. 
\end{proof}

Theorem~\ref{gmips:thm:joint-greedy} guarantees the correctness of 
Algorithm~\ref{gmips:alg:screen} to generate $\cC(\bw)$ based on 
$\pibar(j|\bw)$ defined in \eqref{gmips:eq:greedy-rank}. By 
Theorem~\ref{gmips:thm:bounded-iter}, the overall time complexity of 
Algorithm~\ref{gmips:alg:screen} is $O(Bk\log k)$ as each iteration of the 
{\bf while}-loop costs $O(\log k)$ time.  

The $O(Bk\log k)$ time complexity of Algorithm~\ref{gmips:alg:screen} does not 
satisfy the efficiency requirement of a viable budgeted MIPS approach. Here we 
propose an improved candidate screening procedure which reduces the 
overall time complexity to $O(Bk)$. Observe that the $\log k$ term comes from 
the $\mathtt{Q.push\rbr{\rbr{z_{jt},t}}}$ and $\mathtt{Q.pop()}$ operations of the max-heap for 
each visited $(j,t)$ entry. As the goal of the screening procedure is to identify $j$ 
indices only, we can skip the $\mathtt{Q.push\rbr{\rbr{z_{jt},t}}}$ for an 
entry $(j,t)$ with the $j$ having been included in $\cC(\bw)$. As a result, 
$\mathtt{Q.pop()}$ is executed at most $B+k-1$ times when $\abs{\cC(\bw)}=B$. 
The extra $k-1$ times occurs in the situation that 
\[
  \mathtt{iters}[1]\mathtt{.current()} =  
  \mathtt{iters}[2]\mathtt{.current()} = \cdots = \mathtt{iters}[k]\mathtt{.current()}
\]
at the beginning of the entire screening procedure. 

\begin{algorithm}[t]
  \caption{An improved candidate screening procedure in \gmips. The 
    overall time complexity is $O(Bk)$.}  
\label{gmips:alg:improved-screen}
\begin{itemize}
  \item {\bf Input:} 
    \begin{itemize}
      \item $\cH$, $\bw$, and the computational budget $B$
      \item $\mathtt{Q}$ and $\mathtt{iters}[t]$: output of 
        Algorithm~\ref{gmips:alg:query-dep-proc} 
      \item $\cC(\bw)$: an empty list
      \item $\mathtt{visited}[j]=0,\ j = 1,\ldots,n$: a zero-initialized array of length $n$
    \end{itemize}
  \item {\bf while} $\abs{\cC(\bw)} < B$:
    \begin{itemize}
      \item $(z,t) \leftarrow \mathtt{Q.pop()}$ \hfill $\cdots O(\log k)$
      \item $j \leftarrow \mathtt{iters}[t]\mathtt{.current()}$
      \item {\bf if} $\mathtt{visited}[j] = 0$: 
        \begin{itemize}
          \item append $j$ into $\cC(\bw)$
          \item $\mathtt{visited}[j] \leftarrow 1$
        \end{itemize}
      \item {\bf while} $\mathtt{iters}[t]\mathtt{.hasNext()}$:
        \begin{itemize}
          \item $j \leftarrow \mathtt{iters}[t]\mathtt{.getNext()}$
          \item {\bf if} {$\mathtt{visited}[j] = 0$}:
            \begin{itemize}
              \item $z\leftarrow h_{jt} w_t$
              \item $\mathtt{Q.push}\rbr{\rbr{z,t}}$ \hfill $\cdots O(\log k)$
              \item {\bf break}
            \end{itemize}
        \end{itemize}
    \end{itemize}
  \item $\mathtt{visited}[j]\leftarrow 0, \forall j \in \cC(\bw)$\hfill 
    $\cdots O(B)$
  \item {\bf Output:} $\cC(\bw) = \cbr{j \mid \pibar(j|\bw) \le B}$
\end{itemize}
\end{algorithm}

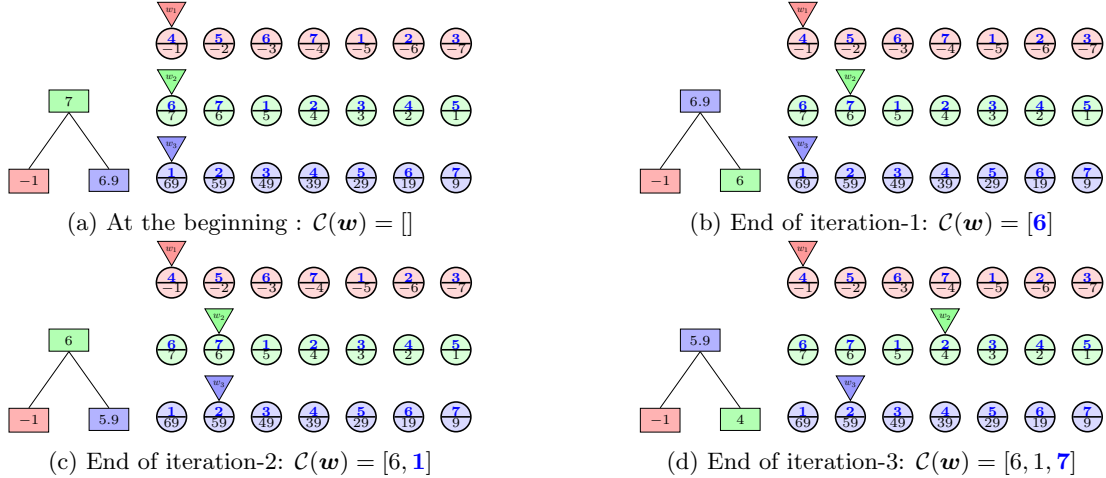
\begin{figure}[!ht]
  \centering
  \begin{subfigure}[t]{0.48\linewidth}
    \centering
      \begin{rescale}[0.75]{0.75}
          \begin{tikzpicture} 
              % draw heap
            \tikzset{level distance=4em, sibling distance=2em}
            %\only<1> 
            {\drawHeap{2}{7}{1}{-1}{3}{6.9}}
            %\only<2> {\drawHeap{3}{6.9}{1}{-1}{2}{6}}
            %\only<3> {\drawHeap{2}{6}{1}{-1}{3}{5.9}}
            %\only<4> {\drawHeap{3}{5.9}{1}{-1}{2}{5}}
          \end{tikzpicture}
          \quad
          \begin{tikzpicture} [
              scale=0.7,
            triangle/.style = {regular polygon, regular polygon sides=3, shape border rotate=180}
            ]
          % draw entire Z
              \foreach \x in {1,...,\rofuK} {    
                \pgfmathparse{\colorPallete[\x-1]};
                \definecolor{curcolor}{rgb}{\pgfmathresult};
                \foreach \y in {1,...,\rofuN} {
                  \pgfmathparse{int(mod(2*\x+\y,\rofuN)+1)}
                  \edef\rofuIdx{\pgfmathresult}
                  \ifthenelse{\x=1}{\pgfmathparse{int(-\y)}}{};
                  \ifthenelse{\x=2}{\pgfmathparse{int(\rofuN-\y+1)}}{};
                  \ifthenelse{\x=3}{\pgfmathparse{int(10*(\rofuN-\y+1)-1)}}{};
                  \edef\rofuVal{\pgfmathresult}
                  \node[circle split, inner sep=0,fill=curcolor!15!white,draw,thick,minimum 
                  width=1.5em] at (1.2*\y-1.2*0.5,1.7*\rofuK-1.7*\x-1.7*0.5) (rofu\x\y) 
                  %{\small \rofuIdx \nodepart{lower} \small {$\rofuVal$}};
                  {\scriptsize \good{\bf \rofuIdx} \nodepart{lower} \scriptsize {$\rofuVal$}};
                }
              };
              %\only<1>
              { \drawWidx{1}{1}{1}}
              %\only<2>{ \drawWidx{1}{2}{1}}
              %\only<3>{ \drawWidx{1}{2}{2}}
              %\only<4>{ \drawWidx{1}{3}{2}}

          \end{tikzpicture}
      \end{rescale}
      \caption{At the beginning : $\cC(\bw) = []$}
    \label{gmips:fig:gmips-1}
  \end{subfigure} \quad
  \begin{subfigure}[t]{0.48\linewidth}
    \centering
      \begin{rescale}[0.75]{0.75}
          \begin{tikzpicture} 
              % draw heap
            \tikzset{level distance=4em, sibling distance=2em}
            %\only<1> {\drawHeap{2}{7}{1}{-1}{3}{6.9}}
            %\only<2> 
            {\drawHeap{3}{6.9}{1}{-1}{2}{6}}
            %\only<3> {\drawHeap{2}{6}{1}{-1}{3}{5.9}}
            %\only<4> {\drawHeap{3}{5.9}{1}{-1}{2}{5}}
          \end{tikzpicture}
          \quad
          \begin{tikzpicture} [
              scale=0.7,
            triangle/.style = {regular polygon, regular polygon sides=3, shape border rotate=180}
            ]
          % draw entire Z
              \foreach \x in {1,...,\rofuK} {    
                \pgfmathparse{\colorPallete[\x-1]};
                \definecolor{curcolor}{rgb}{\pgfmathresult};
                \foreach \y in {1,...,\rofuN} {
                  \pgfmathparse{int(mod(2*\x+\y,\rofuN)+1)}
                  \edef\rofuIdx{\pgfmathresult}
                  \ifthenelse{\x=1}{\pgfmathparse{int(-\y)}}{};
                  \ifthenelse{\x=2}{\pgfmathparse{int(\rofuN-\y+1)}}{};
                  \ifthenelse{\x=3}{\pgfmathparse{int(10*(\rofuN-\y+1)-1)}}{};
                  \edef\rofuVal{\pgfmathresult}
                  \node[circle split, inner sep=0,fill=curcolor!15!white,draw,thick,minimum 
                  width=1.5em] at (1.2*\y-1.2*0.5,1.7*\rofuK-1.7*\x-1.7*0.5) (rofu\x\y) 
                  %{\small \rofuIdx \nodepart{lower} \small {$\rofuVal$}};
                  {\scriptsize \good{\bf \rofuIdx} \nodepart{lower} \scriptsize {$\rofuVal$}};
                }
              };
              %\only<1> { \drawWidx{1}{1}{1}}
              %\only<2>
              { \drawWidx{1}{2}{1}}
              %\only<3>{ \drawWidx{1}{2}{2}}
              %\only<4>{ \drawWidx{1}{3}{2}}

          \end{tikzpicture}
      \end{rescale}
      \caption{End of iteration-$1$: $\cC(\bw) = [\good{\bf 6}]$}
    \label{gmips:fig:gmips-2}
  \end{subfigure}
  \\
  \begin{subfigure}[t]{0.48\linewidth}
    \centering
      \begin{rescale}[0.75]{0.75}
          \begin{tikzpicture} 
              % draw heap
            \tikzset{level distance=4em, sibling distance=2em}
            %\only<1> {\drawHeap{2}{7}{1}{-1}{3}{6.9}}
            %\only<2> {\drawHeap{3}{6.9}{1}{-1}{2}{6}}
            %\only<3> 
            {\drawHeap{2}{6}{1}{-1}{3}{5.9}}
            %\only<4> {\drawHeap{3}{5.9}{1}{-1}{2}{5}}
          \end{tikzpicture}
          \quad
          \begin{tikzpicture} [
              scale=0.7,
            triangle/.style = {regular polygon, regular polygon sides=3, shape border rotate=180}
            ]
          % draw entire Z
              \foreach \x in {1,...,\rofuK} {    
                \pgfmathparse{\colorPallete[\x-1]};
                \definecolor{curcolor}{rgb}{\pgfmathresult};
                \foreach \y in {1,...,\rofuN} {
                  \pgfmathparse{int(mod(2*\x+\y,\rofuN)+1)}
                  \edef\rofuIdx{\pgfmathresult}
                  \ifthenelse{\x=1}{\pgfmathparse{int(-\y)}}{};
                  \ifthenelse{\x=2}{\pgfmathparse{int(\rofuN-\y+1)}}{};
                  \ifthenelse{\x=3}{\pgfmathparse{int(10*(\rofuN-\y+1)-1)}}{};
                  \edef\rofuVal{\pgfmathresult}
                  \node[circle split, inner sep=0,fill=curcolor!15!white,draw,thick,minimum 
                  width=1.5em] at (1.2*\y-1.2*0.5,1.7*\rofuK-1.7*\x-1.7*0.5) (rofu\x\y) 
                  %{\small \rofuIdx \nodepart{lower} \small {$\rofuVal$}};
                  {\scriptsize \good{\bf \rofuIdx} \nodepart{lower} \scriptsize {$\rofuVal$}};
                }
              };
              %\only<1> { \drawWidx{1}{1}{1}}
              %\only<2> { \drawWidx{1}{2}{1}}
              %\only<3> 
              { \drawWidx{1}{2}{2}}
              %\only<4> { \drawWidx{1}{3}{2}}

          \end{tikzpicture}
      \end{rescale}
      \caption{End of iteration-$2$: $\cC(\bw) = [6,\good{\bf 1}]$}
    \label{gmips:fig:gmips-3}
  \end{subfigure} \quad
  \begin{subfigure}[t]{0.48\linewidth}
    \centering
      \begin{rescale}[0.75]{0.75}
          \begin{tikzpicture} 
              % draw heap
            \tikzset{level distance=4em, sibling distance=2em}
            %\only<1> {\drawHeap{2}{7}{1}{-1}{3}{6.9}}
            %\only<2> {\drawHeap{3}{6.9}{1}{-1}{2}{6}}
            %\only<3> {\drawHeap{2}{6}{1}{-1}{3}{5.9}}
            %\only<4> 
            {\drawHeap{3}{5.9}{1}{-1}{2}{4}}
          \end{tikzpicture}
          \quad
          \begin{tikzpicture} [
              scale=0.7,
            triangle/.style = {regular polygon, regular polygon sides=3, shape border rotate=180}
            ]
          % draw entire Z
              \foreach \x in {1,...,\rofuK} {    
                \pgfmathparse{\colorPallete[\x-1]};
                \definecolor{curcolor}{rgb}{\pgfmathresult};
                \foreach \y in {1,...,\rofuN} {
                  \pgfmathparse{int(mod(2*\x+\y,\rofuN)+1)}
                  \edef\rofuIdx{\pgfmathresult}
                  \ifthenelse{\x=1}{\pgfmathparse{int(-\y)}}{};
                  \ifthenelse{\x=2}{\pgfmathparse{int(\rofuN-\y+1)}}{};
                  \ifthenelse{\x=3}{\pgfmathparse{int(10*(\rofuN-\y+1)-1)}}{};
                  \edef\rofuVal{\pgfmathresult}
                  \node[circle split, inner sep=0,fill=curcolor!15!white,draw,thick,minimum 
                  width=1.5em] at (1.2*\y-1.2*0.5,1.7*\rofuK-1.7*\x-1.7*0.5) (rofu\x\y) 
                  %{\small \rofuIdx \nodepart{lower} \small {$\rofuVal$}};
                  {\scriptsize \good{\bf \rofuIdx} \nodepart{lower} \scriptsize {$\rofuVal$}};
                }
              };
              %\only<1> { \drawWidx{1}{1}{1}}
              %\only<2> { \drawWidx{1}{2}{1}}
              %\only<3> { \drawWidx{1}{2}{2}}
              %\only<4> 
              { \drawWidx{1}{4}{2}}
          \end{tikzpicture}
      \end{rescale}
      \caption{End of iteration-$3$: $\cC(\bw) = [6,1,\good{\bf 7}]$}
    \label{gmips:fig:gmips-4}
  \end{subfigure}
  \caption[Illustration of \gmips.]{Illustration of 
    Algorithm~\ref{gmips:alg:improved-screen} with $\bw = [1, 1, 0.1]^\top$ 
    and $B=3$. 
    The left plot for each sub-figure shows the heap structure in the 
    max-heap $\mathtt{Q}$: the value in each rectangle denotes $z$, and each 
    index $t$ is shown in a different color (red for $1$, green for $2$, and 
    blue for $3$). The sorted index arrays are shown in the upper part of 
    circles on the right plot for each sub-figure; for example, 
    $\mathtt{s}_1[4]=7$, $\mathtt{s}_2[1]=6$, and $\mathtt{s}_3[5]=5$. The value in lower part of 
    circles is the corresponding $h_{jt}$; for example, $h_{71}=-4$,
    $h_{62}=7$, and $h_{53}=29$. Three downward triangles denote the current position of 
    $\mathtt{iters}[t],\ t=1,2,3$. Figure~\ref{gmips:fig:gmips-1} shows the 
    status for each data data structure at the beginning of 
    Algorithm~\ref{gmips:alg:improved-screen}. Three pairs are pushed into 
    $\mathtt{Q}$: $(-1=h_{41}w_1, 1)$, $(7=h_{71}w_2, 2)$, and 
    $(6.9=h_{13}w_3, 3)$. 
    Figures~\ref{gmips:fig:gmips-2}-\ref{gmips:fig:gmips-3} show the status 
    in the end of the first and the second iterations of the outer {\bf 
    while}-loop in Algorithm~\ref{gmips:alg:improved-screen}. In 
    Figure~\ref{gmips:fig:gmips-3}, we show that at the third 
    iteration, after $(z,t)=(6,2) \leftarrow \mathtt{Q.pop()}$ is executed and 
    $7=\mathtt{iters}[2]\mathtt{.current()}$ is appended into $\cC(\bw)$, we need to 
    advance $\mathtt{iters[2]}$ twice because the index $j=1$ has been 
    included in $\cC(\bw)$. Note that for this example $\bh_1$ is the 
    candidate with the largest inner product value with $\bw$. 
}
  \label{gmips:fig:gmips-illustration}
\end{figure}
In Algorithm~\ref{gmips:alg:improved-screen}, we give a detailed description for 
this improved candidate screening procedure for \gmips. 
See Figure~\ref{gmips:fig:gmips-illustration} for a detailed illustration of 
this algorithm  on a toy example.  
Note that in Algorithm~\ref{gmips:alg:improved-screen}, we use an auxiliary 
zero-initialized array of length $n$:  
$\mathtt{visited}[j],\ j=1,\ldots,n$ to denote whether an index $j$ has been 
included in $\cC(\bw)$ or not. As $\cC(\bw)$ contains at most $B$ indices, only 
$B$ elements of this auxiliary array will be modified during the screening 
procedure. Furthermore, the auxiliary array can be reset to zero using $O(B)$ 
time in the end of Algorithm~\ref{gmips:alg:improved-screen}, so this 
auxiliary array can be utilized again for a different query vector $\bw$. 

Notice that Algorithm~\ref{gmips:alg:improved-screen} still iterates $Bk$ 
entries of $Z$ but at most $B+k-1$ entries will be pushed into or pop from the 
max-heap. Thus, the overall time complexity of 
Algorithm~\ref{gmips:alg:improved-screen} is $O(Bk + B\log k) = O(Bk)$, which 
satisfies the efficiency requirement for a viable approach for budgeted MIPS. 

{\bf \gmips with a Selection Tree.} As there are at most $k$ pairs in the max-heap $\mathtt{Q}$, one 
from each $\mathtt{iters}[t]$, the max-heap can be replaced by a selection 
tree to achieve a slightly faster implementation as suggested in \cite[Chapter 
5.4.1]{TAOCP3}. In Algorithm~\ref{gmips:alg:selection-tree}, we give a pseudo 
code for the selection tree with a $O(k)$ time constructor, a $O(1)$ time 
maximum element look-up, and a $O(\log k)$ time updater. To apply the section 
tree for our \gmips, we only need to the following modifications:
\begin{itemize}
  \item In Algorithm~\ref{gmips:alg:query-dep-proc}, remove 
    $\mathtt{Q.push}((z,t))$ from the {\bf for}-loop and construct $\mathtt{Q}$ 
    by $\mathtt{Q} \leftarrow \mathtt{SelectionTree}(\bw, k, \mathtt{iters})$.
  \item In Algorithm~\ref{gmips:alg:joint-iter} and 
  Algorithm~\ref{gmips:alg:improved-screen}, replace $\mathtt{Q.pop()}$ by 
  $\mathtt{Q.top()}$ and replace $\mathtt{Q.push}((z,t))$ by 
  $\mathtt{Q.updateValue}(t, z)$. 
\end{itemize}
\begin{algorithm}[t]
  \caption{A pseudo code of a selection tree used for \gmips.}
  \label{gmips:alg:selection-tree}
  \begin{itemize}
    \item[] {\bf class} {\tt SelectionTree:} 
      \vspace{-0.8em}
      \begin{itemize}
        \item[]  {\bf def} $\mathtt{constructor}(\bw, k,\mathtt{iters}):$\hfill $\cdots O(k)$
          \begin{itemize}
            \item[] $\Kbar \leftarrow \min \cbr{2^i\mid 2^i \ge k}$
            \item[] {\bf for} ${i}=1,\ldots,2\Kbar$:
              \begin{itemize}
                \item[] $\mathtt{buf}[i] \leftarrow (-\infty,0)$
              \end{itemize}
            \item[] {\bf for} $t=1,\ldots,k$:
              \begin{itemize}
                \item[] $j \leftarrow \mathtt{iters}[t]\mathtt{.current()}$
                \item[] $\mathtt{buf}[\Kbar+t] \leftarrow (h_{jt}w_t,t)$
              \end{itemize}
            \item[] {\bf for} $i=\Kbar,\ldots,1$: 
              \begin{itemize}
                \item[] {\bf if} 
                  $\mathtt{buf}[2i]\mathtt{.first} > \mathtt{buf}[2i+1]\mathtt{.first}$: 
                \item[] $\qquad\mathtt{buf}[i] \leftarrow \mathtt{buf}[2i]$
                \item[] {\bf else}: 
                \item[] $\qquad\mathtt{buf}[i] \leftarrow \mathtt{buf}[2i+1]$
              \end{itemize}
          \end{itemize}
        \item[] {\bf def} $\mathtt{top}()$: {\bf return} $\mathtt{buf}[1]$ 
          \hfill $\cdots O(1)$
        \item[] {\bf def} $\mathtt{updateValue}(t, z)$: \hfill $\cdots O(\log k)$
          \begin{itemize}
            \item[] $i \leftarrow \Kbar + t$
            \item[] $\mathtt{buf}[i] \leftarrow (z, t)$
            \item[] {\bf while} $i > 1$:
              \begin{itemize}
                \item[] $i \leftarrow \floor{i/2}$
                \item[] {\bf if} 
                  $\mathtt{buf}[2i]\mathtt{.first} > \mathtt{buf}[2i+1]\mathtt{.first}$: 
                \item[] $\qquad\mathtt{buf}[i] \leftarrow \mathtt{buf}[2i]$
                \item[] {\bf else}: 
                \item[] $\qquad\mathtt{buf}[i] \leftarrow \mathtt{buf}[2i+1]$
              \end{itemize}
          \end{itemize}
      \end{itemize}
  \end{itemize}
  
\end{algorithm}

\subsubsection{Connection to Sampling-based MIPS Approaches}
\label{gmips:sec:discussion}
\smips, as mentioned earlier, is essentially a sampling algorithm with replacement scheme 
to draw entries of $Z$ such that $(j,t)$ is sampled with the probability proportional 
to $ z_{jt}$. Thus, \smips can be thought as a traversal of $(j,t)$ entries 
using in a stratified random sequence determined by a distribution of the {\em  values} of 
$\cbr{z_{jt}}$, while the core idea of \gmips is to iterate $(j,t)$ entries of  
$Z$ in a greedy sequence induced by the {\em ordering} of $\cbr{z_{jt}}$. 
Next, we discuss the differences between \gmips and \smips in a few 
perspectives: 

\paragraph{Applicability:} \smips can be applied to the situation where both 
$\cH$ and $\bw$ are nonnegative because of the nature of sampling scheme. In 
contrast, \gmips can work on any MIPS problems as only the ordering of 
$\cbr{z_{jt}}$ matters in \gmips. Instead of $\bh_j^\top\bw$, \dmips
is designed for the MSIPS problem which is to identify candidates 
with largest $\rbr{\bh_j^\top\bw}^2$ or $\abs{\bh_j^\top\bw}$ 
values~\cite{GB15a}. In fact, for nonnegative MIPS problems, the diamond 
sampling is equivalent to \smips. Moreover, for MSIPS problems with negative 
entries, when the number of samples is set to be the budget $B$,\footnote{This 
setting is used in the experiments in \cite{GB15a}} the \dmips
is equivalent to apply \smips to sample $(j,t)$  entries with the 
probability $p(j,t)\propto \abs{z_{jt}}$. Thus, the applicability of the 
existing sampling-based approaches is still very limited for general MIPS 
problems. 

\paragraph{Flexibility to Control $\abs{\cC(\bw)}$:} By 
Theorem~\ref{gmips:thm:bounded-iter}, we know that \gmips can guarantee both 
the time complexity of the candidate screening procedure and the size of 
output $\abs{\cC(\bw)}$ for any $\cH$, $\bw$, and $B$. For a sampling-based 
approach, one can easily control either the time complexity of the sampling 
procedure or the size of $\cC(\bw)$, but not both. Because all the existing 
sampling-based approaches are a sampling scheme with replacement, the same 
entry $(j,t)$ could be sampled repeatedly. Thus, the time complexity
to guarantee that $\cC(\bw)=B$ depends on the distribution of values of $\bw$ 
and $\cH$. Hence, \gmips is more flexible than sampling-based approaches in 
terms of the controllability of $\cC(\bw)$.

\section{Experimental Results}
\label{gmips:sec:exp}
In this section, we perform extensive empirical comparisons to compare \gmips 
with other state-of-the-art fast MIPS approaches on both real-world and 
synthetic datasets. 
\begin{itemize}
  \item We use \netflix and \yahoo as our real-world recommender 
  system datasets. There are $17,770$ and $624,961$ items in \netflix and 
  \yahoo, respectively. In particular, we obtain the low rank model $(W,H)$ by the 
  standard low-rank matrix factorization: 
  \[
    \min_{W,H} \sum_{(i,j)\in \Omega}\ \rbr{A_{ij} - \bw_i^\top\bh_j}^2 + 
    \lambda\rbr{\sum_{i=1}^m \frac{\abs{\Omega_i}}{n} \norm{\bw_i}^2 + 
    \sum_{j=1}^n\frac{\abs{\Omegabar_j}}{m} \norm{\bh_j}^2},
  \]
  where $A_{ij}$ is the rating of the $j$-th item given by the $i$-th user, 
  $\Omega$ is the set of observed ratings, $\Omega_i=\cbr{j\mid (i,j)\in\Omega}$, 
  and $\Omegabar_j = \cbr{i\mid (i,j)\in \Omega}$, and $\lambda$ is a 
  regularization parameter. We use the CCD++~\cite{HFY14a} algorithm implemented in 
  LIBPMF\footnote{\url{http://www.cs.utexas.edu/~rofuyu/libpmf}} to solve the  
  above optimization problem and obtain the user embeddings $\cbr{\bw_i}$ and 
  item embeddings $\cbr{\bh_j}$. We use the same $\lambda$ used 
  in~\cite{WSC15b}. We also obtain $(W,H)$ with a different $k$: $50$, $100$, 
  and $k=200$. 
\item We also generate synthetic datasets with various 
  $n = 2^{\cbr{17,18,19,20}}$ and $k=2^{\cbr{2,5,7,10}}$. For each synthetic 
  dataset, both candidate vector $\bh_j$ and query $\bw$ vector are drawn from 
  the normal distribution. 
\end{itemize}

\begin{figure}[!t]
  \centering
  \caption{Comparison of variants of \gmips. }
  \label{gmips:fig:gmips-comp}
    \begin{resize}{1\linewidth}
      \begin{tabular}{@{}c@{}c@{}c@{}c@{}}
%        %\hspace{-4\leftshift}
%        \begin{subfigure}[h]{0.35\linewidth}
%          \includegraphics[width=1\linewidth]{figs/netflix50-greedy-comp-p_5.png}
%        \end{subfigure}
%        \hspace{-\leftshift}
%        \begin{subfigure}[h]{0.35\linewidth}
%          \includegraphics[width=1\linewidth]{figs/netflix50-greedy-comp-p_10.png}
%        \end{subfigure}
%        \hspace{-\leftshift}
%        \begin{subfigure}[h]{0.35\linewidth}
%          \includegraphics[width=1\linewidth]{figs/yahoo50-greedy-comp-p_5.png}
%        \end{subfigure}
%        \hspace{-\leftshift}
%        \begin{subfigure}[h]{0.35\linewidth}
%          \includegraphics[width=1\linewidth]{figs/yahoo50-greedy-comp-p_10.png}
%        \end{subfigure} 
%        \\%\hspace{-4\leftshift}
%        \begin{subfigure}[h]{0.35\linewidth}
%          \includegraphics[width=1\linewidth]{figs/netflix-greedy-comp-p_5.png}
%        \end{subfigure}
%        \hspace{-\leftshift}
%        \begin{subfigure}[h]{0.35\linewidth}
%          \includegraphics[width=1\linewidth]{figs/netflix-greedy-comp-p_10.png}
%        \end{subfigure}
%        \hspace{-\leftshift}
%        \begin{subfigure}[h]{0.35\linewidth}
%          \includegraphics[width=1\linewidth]{figs/yahoo-greedy-comp-p_5.png}
%        \end{subfigure}
%        \hspace{-\leftshift}
%        \begin{subfigure}[h]{0.35\linewidth}
%          \includegraphics[width=1\linewidth]{figs/yahoo-greedy-comp-p_10.png}
%        \end{subfigure}
%        \\% \hspace{-4\leftshift}
        %\hspace{-\leftshift}
        \begin{subfigure}[h]{0.35\linewidth}
          \includegraphics[width=1\linewidth]{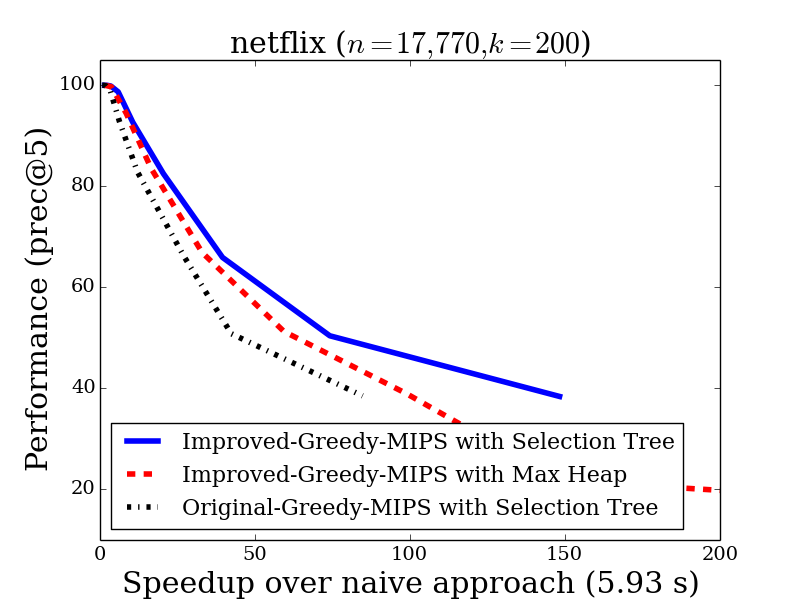}
        \end{subfigure}
        \hspace{-\leftshift}
        \begin{subfigure}[h]{0.35\linewidth}
          \includegraphics[width=1\linewidth]{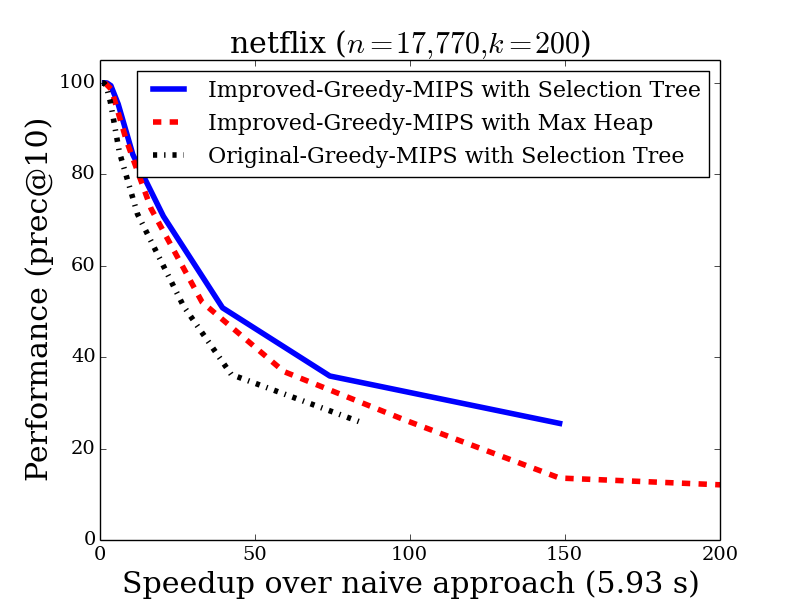}
        \end{subfigure}
        \\ %\hspace{-\leftshift} 
      \begin{subfigure}[h]{0.35\linewidth}
        \includegraphics[width=1\linewidth]{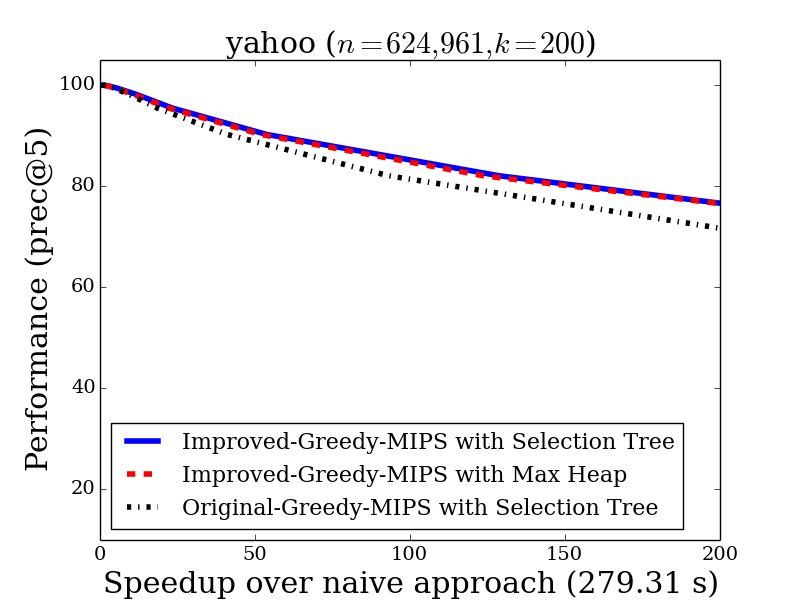}
      \end{subfigure}
        \hspace{-\leftshift}
        \begin{subfigure}[h]{0.35\linewidth}
          \includegraphics[width=1\linewidth]{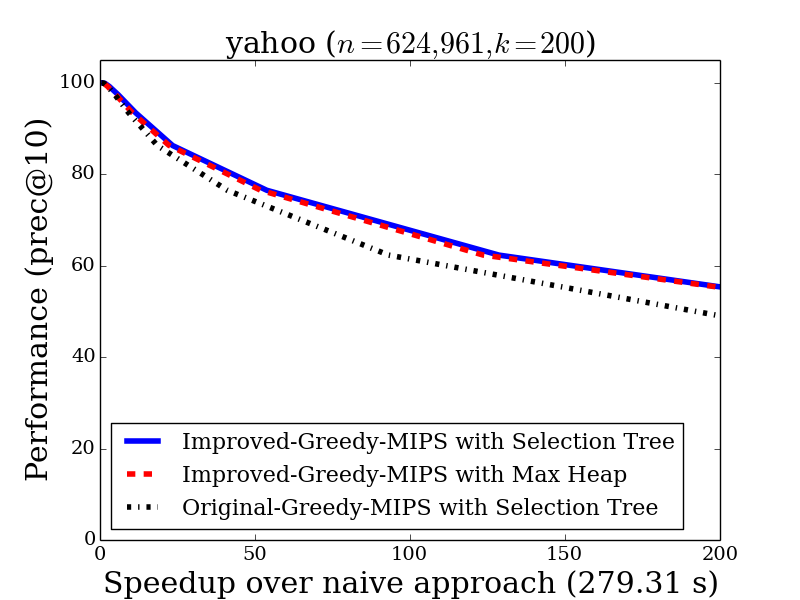}
      \end{subfigure}
      \end{tabular}
    \end{resize}
\end{figure}

\begin{figure}[!th]
  \centering
    %\vspace{-2em}
    \begin{resize}{1\linewidth}
      \begin{tabular}{@{}c@{}}
        \begin{subfigure}[h]{0.35\linewidth}
          \includegraphics[width=1\linewidth]{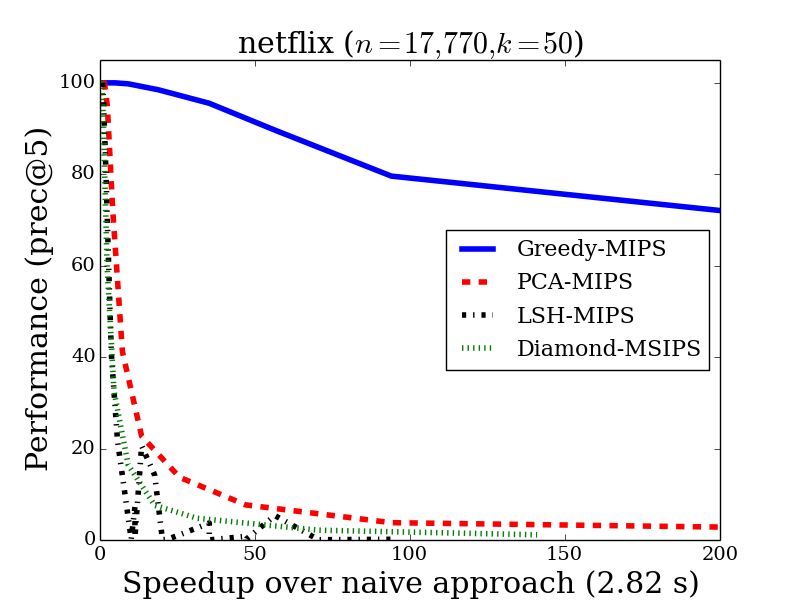}
        \end{subfigure}
        \hspace{-\leftshift}
        \begin{subfigure}[h]{0.35\linewidth}
          \includegraphics[width=1\linewidth]{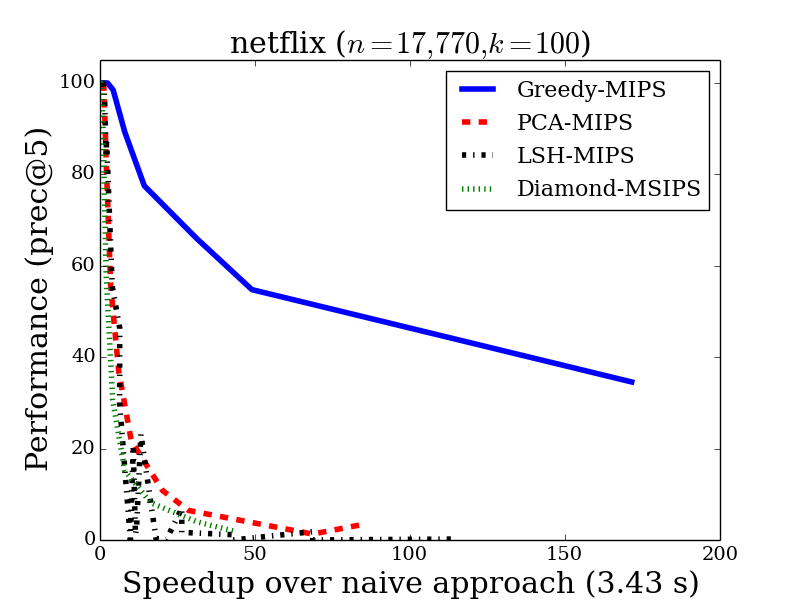}
        \end{subfigure}
        \hspace{-\leftshift}
        \begin{subfigure}[h]{0.35\linewidth}
          \includegraphics[width=1\linewidth]{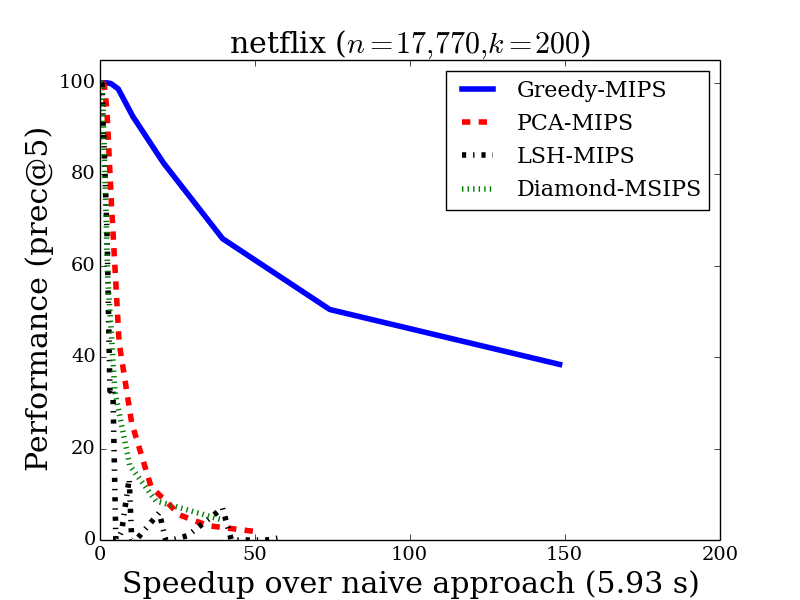}
        \end{subfigure}
        \\
        \begin{subfigure}[h]{0.35\linewidth}
          \includegraphics[width=1\linewidth]{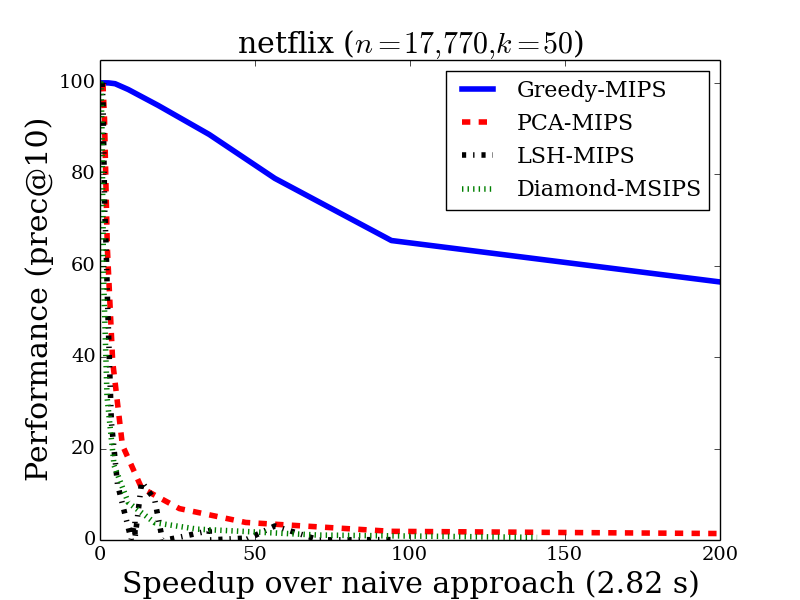}
        \end{subfigure}
        \hspace{-\leftshift}
        \begin{subfigure}[h]{0.35\linewidth}
          \includegraphics[width=1\linewidth]{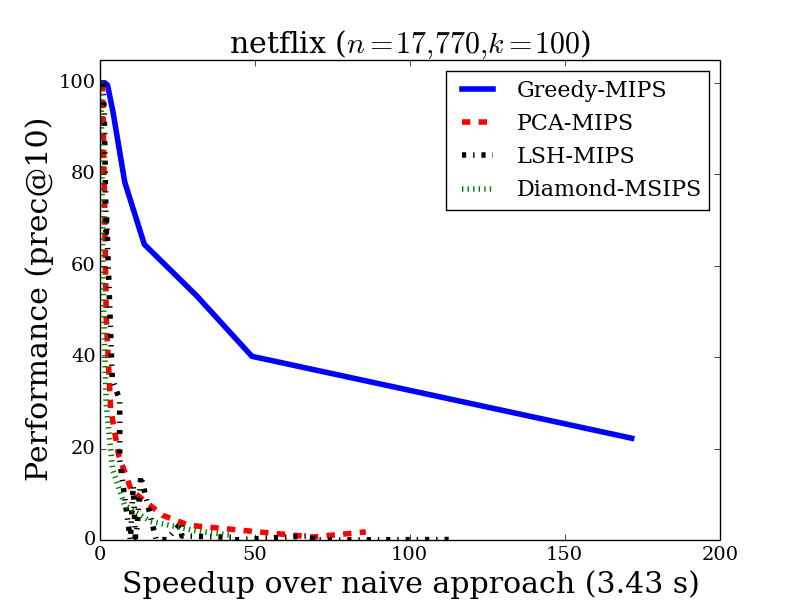}
        \end{subfigure}
        \hspace{-\leftshift}
        \begin{subfigure}[h]{0.35\linewidth}
          \includegraphics[width=1\linewidth]{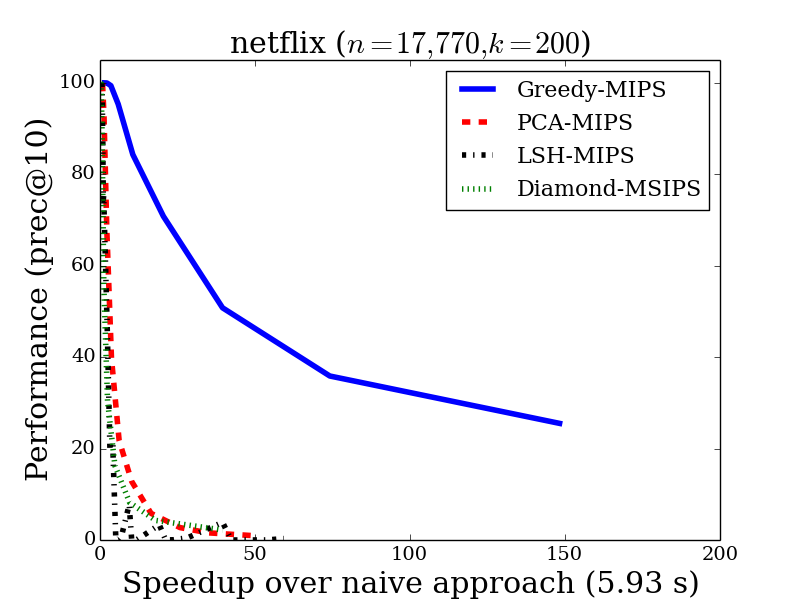}
        \end{subfigure}
        \\
        \\
        \hline
        \hline
        \begin{subfigure}[h]{0.35\linewidth}
          \includegraphics[width=1\linewidth]{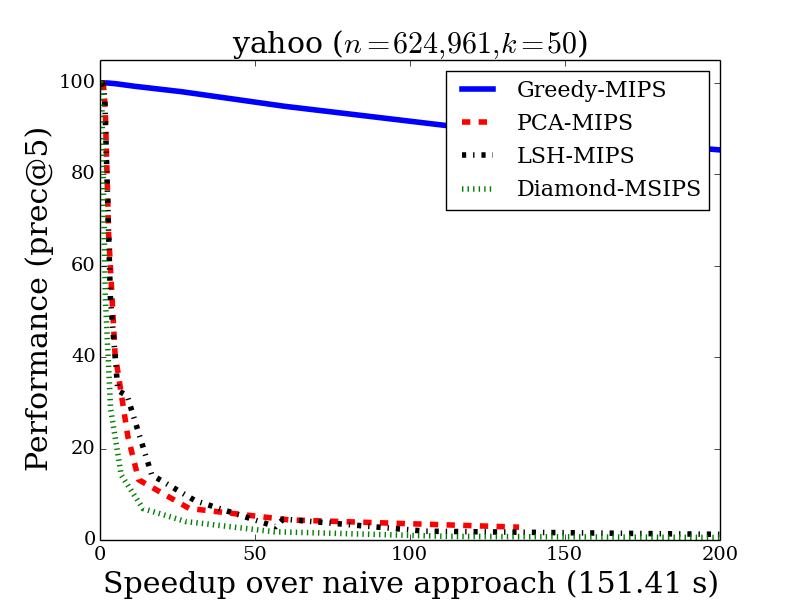}
        \end{subfigure}
        \hspace{-\leftshift}
        \begin{subfigure}[h]{0.35\linewidth}
          \includegraphics[width=1\linewidth]{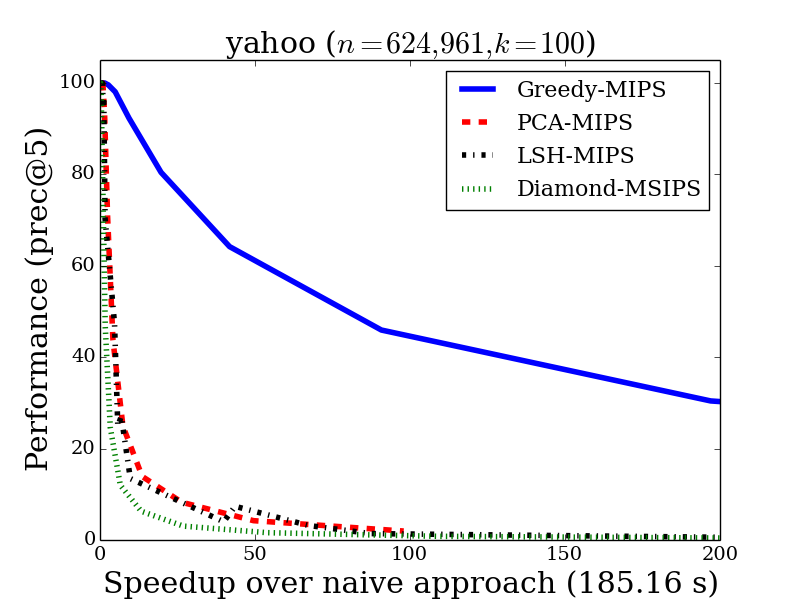}
        \end{subfigure}
          \hspace{-\leftshift}
        \begin{subfigure}[h]{0.35\linewidth}
          \includegraphics[width=1\linewidth]{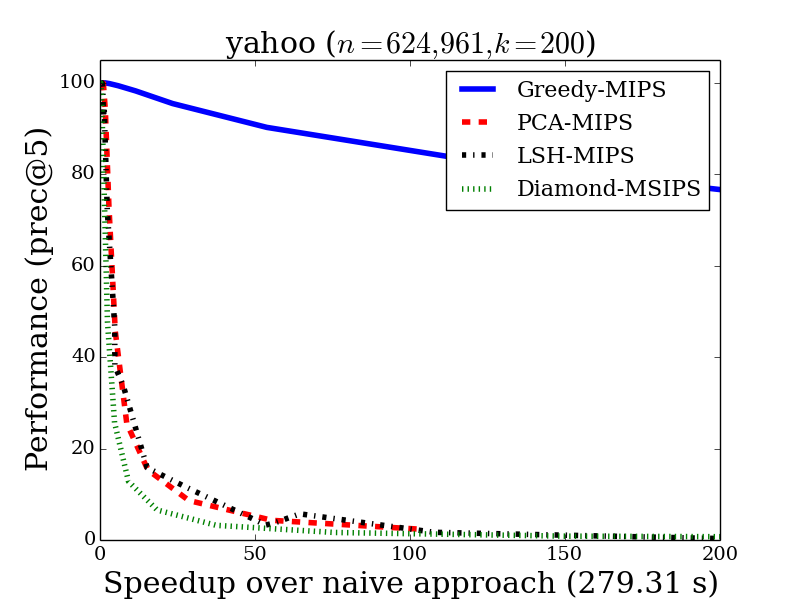}
        \end{subfigure}
        \\
        \begin{subfigure}[h]{0.35\linewidth}
          \includegraphics[width=1\linewidth]{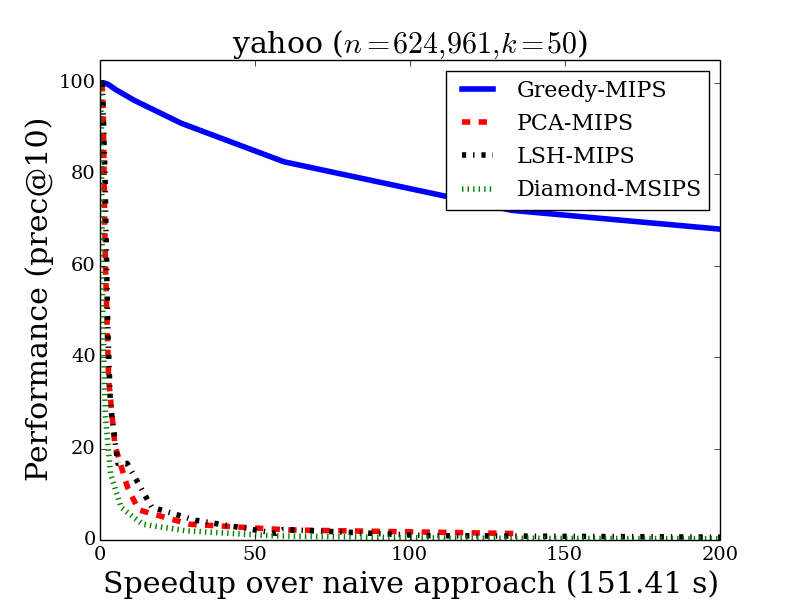}
        \end{subfigure} 
        \hspace{-\leftshift}
        \begin{subfigure}[h]{0.35\linewidth}
          \includegraphics[width=1\linewidth]{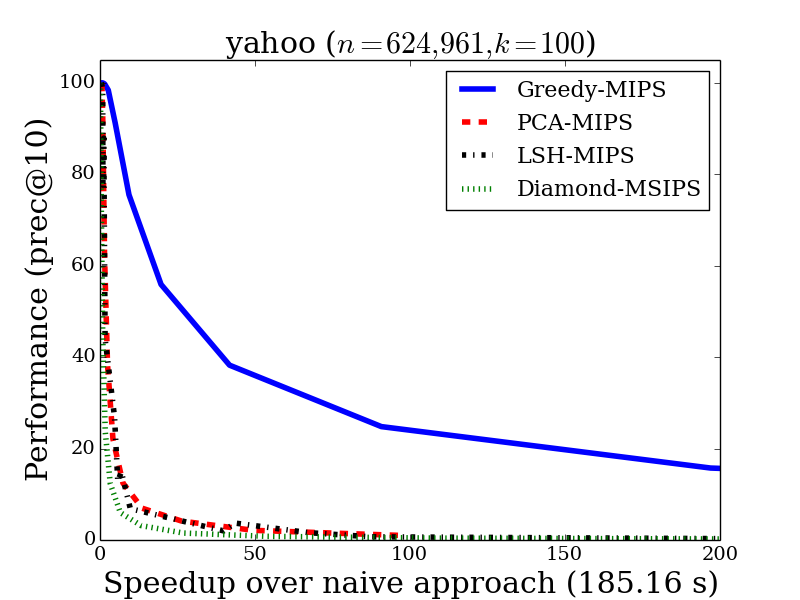}
        \end{subfigure}
        \hspace{-\leftshift}
        \begin{subfigure}[h]{0.35\linewidth}
          \includegraphics[width=1\linewidth]{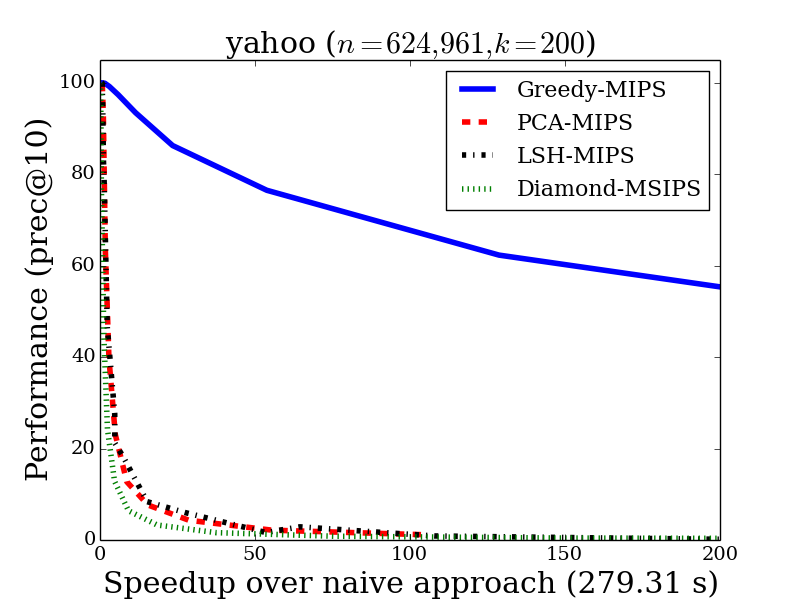}
        \end{subfigure}
      \end{tabular}
    \end{resize}
  \caption{MIPS Comparison on \netflix and \yahoo. }
  \label{gmips:fig:realworld-comp}
\end{figure}

\begin{figure}[!th]
  \centering
    %\vspace{-2em}
  \begin{tabular}{@{}c@{}c@{}}
    %\hspace{-4\leftshift}
    \begin{subfigure}[h]{0.45\linewidth}
      \includegraphics[width=1\linewidth]{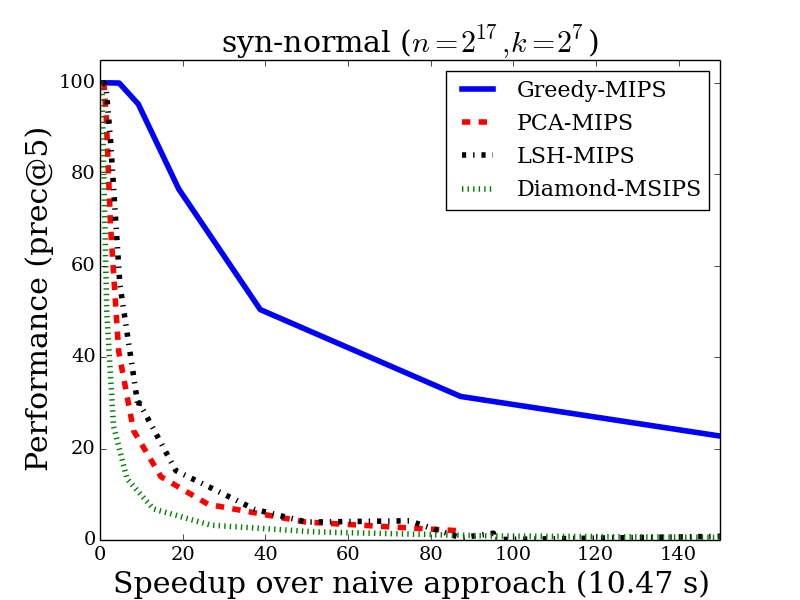}
    \end{subfigure}
    %\hspace{-\leftshift}
    \begin{subfigure}[h]{0.45\linewidth}
      \includegraphics[width=1\linewidth]{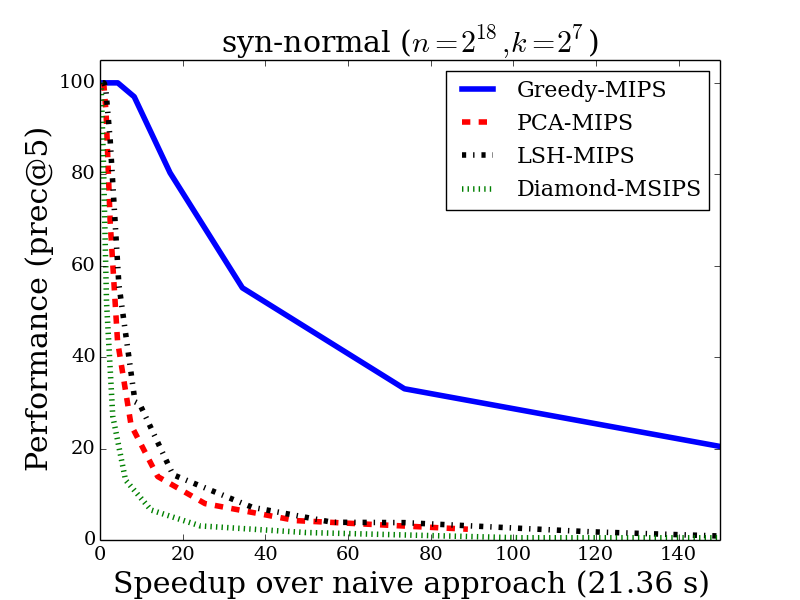}
    \end{subfigure} \\
    %\hspace{-\leftshift}
    \begin{subfigure}[h]{0.45\linewidth}
      \includegraphics[width=1\linewidth]{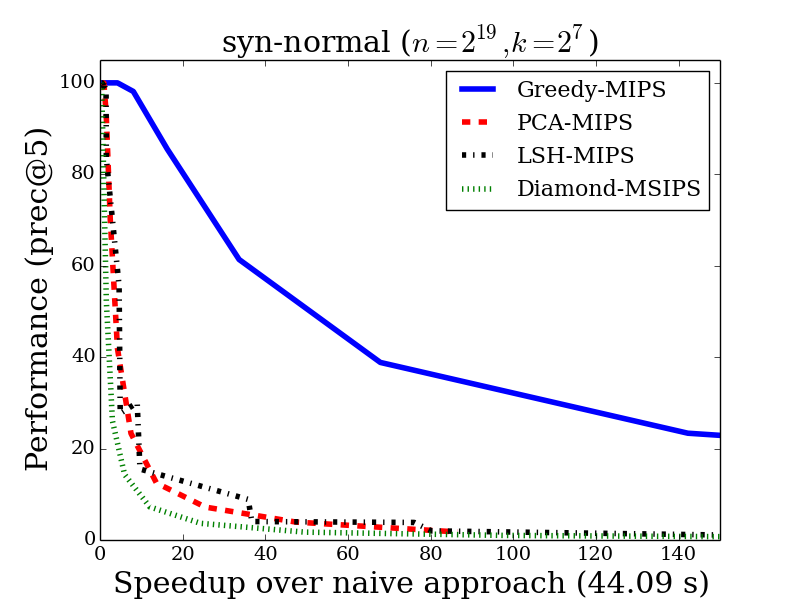}
    \end{subfigure}
    %\hspace{-\leftshift}
    \begin{subfigure}[h]{0.45\linewidth}
      \includegraphics[width=1\linewidth]{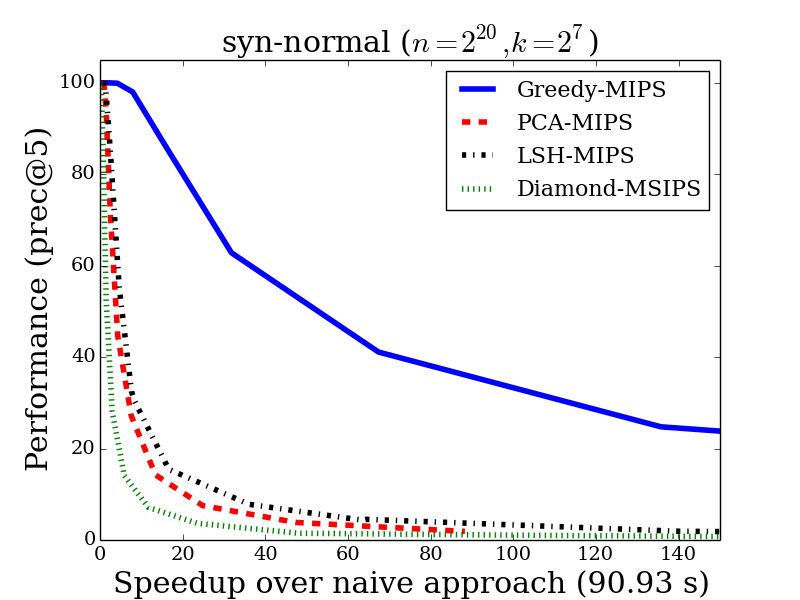}
    \end{subfigure}
    \begin{comment}
    \\ \hspace{-4\leftshift}
    \begin{subfigure}[h]{0.35\linewidth}
      \includegraphics[width=1\linewidth]{figs/pos.m17.d7-comp-p_5.png}
    \end{subfigure}
    \hspace{-\leftshift}
    \begin{subfigure}[h]{0.35\linewidth}
      \includegraphics[width=1\linewidth]{figs/pos.m18.d7-comp-p_5.png}
    \end{subfigure}
    \hspace{-\leftshift}
    \begin{subfigure}[h]{0.35\linewidth}
      \includegraphics[width=1\linewidth]{figs/pos.m19.d7-comp-p_5.png}
    \end{subfigure}
    \hspace{-\leftshift}
    \begin{subfigure}[h]{0.35\linewidth}
      \includegraphics[width=1\linewidth]{figs/pos.m20.d7-comp-p_5.png}
    \end{subfigure} \\
    \end{comment}
  \end{tabular}
  \caption[MIPS Comparison on synthetic datasets with various $m$.]{MIPS Comparison on synthetic datasets with 
  $n \in 2^{\cbr{17,18,19,20}}$ and $k =128$. The datasets used to generate 
  results are created with each entry drawn from a normal 
  distribution. %The datasets for the bottom row are obtained with each entry 
  %drawn from a uniform distribution over $[0,1]$.  
   } 
  \label{gmips:fig:various_m}
\end{figure}

\subsection{Experimental Settings and Evaluation Criteria}
All the experiments are performed on a Linux machine with 20 cores and 256 GB 
memory.  We ensure that only single core/thread is used for our experiments. 
To have a fair comparison, all the compared approaches are implemented in C++: 
\begin{itemize}
  \item \gmips: our proposed approach in Section~\ref{gmips:sec:greedy-mips}. 
    We compare the following variants in Section~\ref{gmips:sec:exp-result}:
    \begin{itemize}
      \item The improved \gmips in Algorithm~\ref{gmips:alg:improved-screen} 
        with the selection tree in Algorithm~\ref{gmips:alg:selection-tree}.   
      \item The improved \gmips in Algorithm~\ref{gmips:alg:improved-screen} 
        with a max-heap.  
      \item The original \gmips in Algorithm~\ref{gmips:alg:screen} with the 
        selection tree in Algorithm~\ref{gmips:alg:selection-tree}. 
    \end{itemize}
  \item NNS-based MIPS approaches: 
    \begin{itemize}
      \item \pmips: the approach proposed in~\cite{YB14a}, which is shown to be 
        the state-of-the-art among tree-based approaches~\cite{YB14a}. We 
        implement a complete PCA-Tree with the neighborhood boosting techniques 
        described in \cite{YB14a}. We vary the depth of PCA 
        tree to control the trade-off between the search quality and the search 
        efficiency. 
      \item \lmips: the approach proposed in~\cite{AS14a,BN15a}. We use the 
        nearest neighbor transform function proposed in ~\cite{YB14a,BN15a} and 
        use the random projection scheme as the LSH function as suggested in ~\cite{BN15a}. 
        We also implement the standard amplification procedure with an OR-construction of $b$ 
        hyper LSH hash functions. Each hyper LSH function is a result of an 
        AND-construction of $a$ random projections. We vary the values $(a,b)$ 
        to control the trade-off between the search quality and the search 
        efficiency. 
    \end{itemize}
%  \item \smips: the approach proposed in~\cite{EC97a,GB15a}. 
%    F+Tree~\cite{HFY15a} is implemented as the multinomial sampler for each 
%    query-independent conditional distribution (See details in 
%    Section~\ref{gmips:sec:related}). Note that \smips only works for 
%    nonnegative MIPS problems. 
  \item \dmips: the sampling scheme proposed in~\cite{GB15a} for the 
    maximum squared inner product search. As it shows better 
    performance than \lmips in \cite{GB15a} in terms of MIPS problems, we also 
    include \dmips into our comparison. F+Tree~\cite{HFY15a} is implemented as 
    the multinomial sampling procedure. 
 \item \nmips: the baseline approach which applies a linear search to identify 
   the exact top-$K$ candidates. 
\end{itemize}

{\bf Evaluation Criteria.}
For each dataset, the actual top-20 items for each query are regarded as 
the ground truth. We report the average performance on a randomly selected 
2,000 query vectors. To evaluate the search quality,  we use the precision on 
the top-$K$ prediction (prec$@K$), is obtained by selecting top-$K$ items from 
$\cC(\bw)$ returned by the candidate screening procedure of a compared MIPS 
approach. $K=5$ and $K=10$ are used in our experiments. To evaluate the search efficiency, we 
report the relative speedups over the \nmips approach as follows: 
\[
  \text{speedup} = \frac{\text{prediction time required by 
  \nmips}}{\text{prediction time by a compared approach}}. 
\]

{\bf Remarks on Budgeted MIPS versus Non-Budgeted MIPS.} As mentioned in 
Section~\ref{gmips:sec:bmips}, \pmips and \lmips cannot handle MIPS with a budget. 
Both the search computation cost and the search quality are fixed when the 
corresponding data structure is constructed. As a result, to understand the 
trade-off between search efficiency and search quality for these two approaches,
we can only try various values for its parameters (such as the depth for PCA 
tree and the amplification parameters $(a,b)$ for LSH). For each combination 
of parameters, we need to re-run the entire query-independent pre-processing 
procedure to construct a new data structure.

\begin{figure}[!th]
  \centering
    %\vspace{-2em}
  \begin{tabular}{@{}c@{}c@{}}
    %\hspace{-4\leftshift}
    \begin{subfigure}[h]{0.45\linewidth}
      \includegraphics[width=1\linewidth]{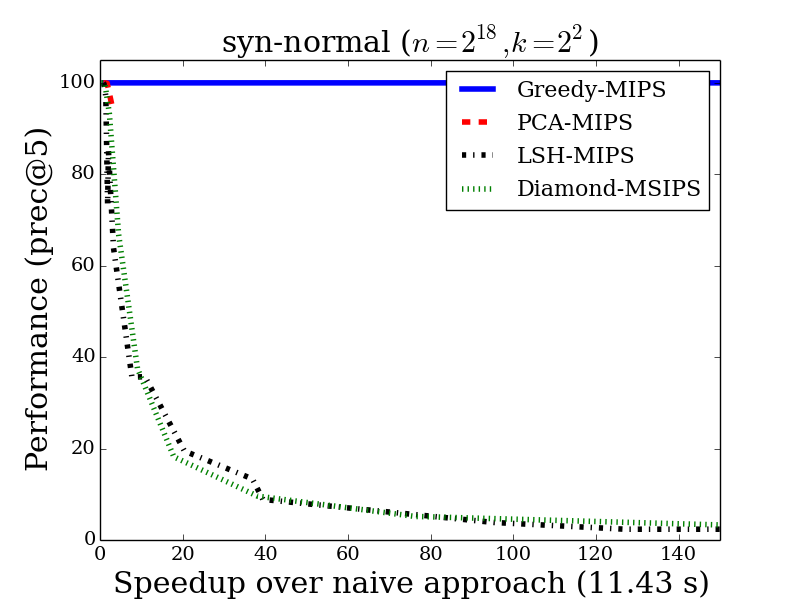}
    \end{subfigure}
    %\hspace{-\leftshift}
    \begin{subfigure}[h]{0.45\linewidth}
      \includegraphics[width=1\linewidth]{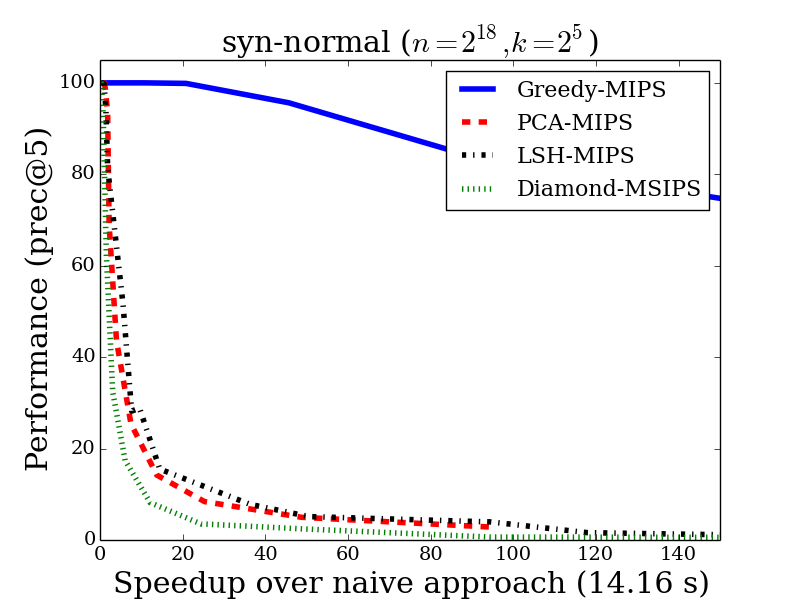}
    \end{subfigure}
    \\
    %\hspace{-\leftshift}
    \begin{subfigure}[h]{0.45\linewidth}
      \includegraphics[width=1\linewidth]{figs/syn.m18.d7-comp-p_5.png}
    \end{subfigure}
    %\hspace{-\leftshift}
    \begin{subfigure}[h]{0.45\linewidth}
      \includegraphics[width=1\linewidth]{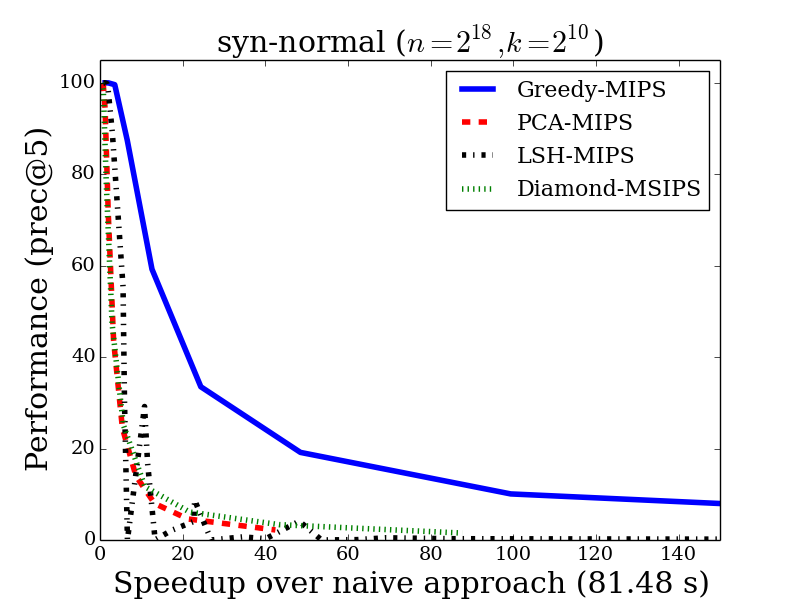}
    \end{subfigure} 
    \begin{comment}
    \\ \hspace{-4\leftshift}
    \begin{subfigure}[h]{0.35\linewidth}
      \includegraphics[width=1\linewidth]{figs/pos.m18.d2-comp-p_5.png}
    \end{subfigure}
    \hspace{-\leftshift}
    \begin{subfigure}[h]{0.35\linewidth}
      \includegraphics[width=1\linewidth]{figs/pos.m18.d5-comp-p_5.png}
    \end{subfigure}
    \hspace{-\leftshift}
    \begin{subfigure}[h]{0.35\linewidth}
      \includegraphics[width=1\linewidth]{figs/pos.m18.d7-comp-p_5.png}
    \end{subfigure}
    \hspace{-\leftshift}
    \begin{subfigure}[h]{0.35\linewidth}
      \includegraphics[width=1\linewidth]{figs/pos.m18.d10-comp-p_5.png}
    \end{subfigure} \\
    \end{comment}
  \end{tabular}
  \caption[MIPS Comparison on synthetic datasets with various $k$.]{MIPS Comparison on synthetic datasets with 
    $n = 2^{18}$ and $k \in 2^{\cbr{2,5,7,10}}$. The datasets used to generate 
  results on are created with each entry drawn from a normal 
  distribution. %The datasets for the bottom row are obtained with each entry 
  %drawn from a uniform distribution over $[0,1]$. 
  } 
  \label{gmips:fig:various_k}
\end{figure}

\subsection{Experimental Results}
\label{gmips:sec:exp-result}
{\bf Results on Variants of \gmips.} In Figure~\ref{gmips:fig:gmips-comp}, we 
shows the comparison between the three variants of \gmips on \netflix and 
\yahoo. We can see that the difference between the use of a selection-tree and 
a max-heap is small, while the different between the use of 
Algorithm~\ref{gmips:alg:screen} and the use of 
Algorithm~\ref{gmips:alg:improved-screen} is more significant. For the 
comparison to other MIPS approaches, we use \gmips to denote the results 
obtained from the version with the combination of 
Algorithm~\ref{gmips:alg:improved-screen} and  
Algorithm~\ref{gmips:alg:selection-tree}. 

{\bf Results on Real-World Data sets.} Comparison results for \netflix and 
\yahoo are shown in Figure~\ref{gmips:fig:realworld-comp}. The first, second, and
third columns present the results with $k=50$, $k=100$, and $k=200$, 
respectively. It is clearly observed that given a fixed speedup, \gmips yields 
predictions with much higher search quality. In particular, on the \yahoo data 
set with $k=200$, \gmips runs 200x  faster than \nmips and yields search 
results with $p_5=70\%$, while none of \pmips, \lmips, and \dmips can achieve a 
$p_5 > 10\%$ while maintaining the similar 200x speedups.   

{\bf Results on Synthetic Data Sets.} We also perform comparisons on synthetic 
datasets. The comparison with various $n\in2^{\cbr{17,18,19,20}}$ is shown in 
Figure~\ref{gmips:fig:various_m}, while the comparison with various 
$k \in 2^{\cbr{2,5,7,10}}$  is shown in Figure~\ref{gmips:fig:various_k}. We 
observe that the performance gap between \gmips over other approaches remains 
when $n$ increases, while the gap becomes smaller when $k$ increases. However, 
\gmips still outperforms other approaches significantly. 

%{\bf Results on }

\begin{comment}
\begin{itemize}
  \item various $n$
  \item various $k$
  \item normal matrix v.s. nonnegative matrix 
  \item two real-world datasets
  \item detailed comparison to sampling approaches ?
  %\item Speed Comparison
  %\item budget v.s. performance 
  %\item transformation
\end{itemize}
\end{comment}

%\section{Discussions and Future Works}
%\label{gmips:sec:ext}
%
%\begin{itemize}
%  \item Sparse Candidates 
%  \item Extension to Greedy Coordinate Descent
%  \item Worst Case Scenario 
%\end{itemize}
%
%Citations
%    \begin{itemize}
%      \item \smips: \cite{EC97a,GB15a}
%      \item 
%    \end{itemize}

%\begin{comment}
\section{Conclusions}
\label{gmips:sec:conclusions}

In this paper, we study the computational issue in the  
prediction phase for many MF-based models: a maximum inner product search problem 
(MIPS) with a very large number of candidate embeddings. By carefully 
studying the problem structure of MIPS, we develop a novel \gmips 
algorithm, which can handle budgeted MIPS by design. While simple and 
intuitive, \gmips yields surprisingly superior performance compared to 
state-of-the-art approaches. As a specific example, on a candidate set 
containing half a million vectors of dimension 200, Greedy-MIPS runs 200x 
faster than the naive approach while yielding search results with the top-5
precision greater than 75\%.       

\small
\bibliographystyle{plain}
\bibliography{rf}
\end{document}